\documentclass[runningheads]{llncs}
\usepackage[T1]{fontenc}
\usepackage{multirow}
\usepackage{graphicx}
\usepackage{booktabs}
\usepackage{mathrsfs}
\newtheorem{thm}{Theorem}
\usepackage[T1]{fontenc}
\usepackage{mathtools}
\usepackage{xcolor}
\usepackage{float}
\usepackage{amssymb}
\usepackage{pifont}
\newcommand{\cmark}{\ding{51}}%
\newcommand{\xmark}{\ding{55}}
\begin{document}
\title{A Sinkhorn Regularized Adversarial Network for Image Guided DEM Super-resolution using Frequency Selective Hybrid Graph Transformer}
%
\titlerunning{Guided DEM SR using Sinkhorn Regularized Adversarial Graph Transformer}
%
\author{Subhajit Paul\inst{1} \and
Ashutosh Gupta\inst{1}}
\authorrunning{S. Paul et al.}
%
\institute{Space Applications Centre (SAC), ISRO, Ahmedabad, India\\
\email{\{subhajitpaul,ashutoshg\}@sac.isro.gov.in}}
\maketitle              
\begin{abstract}
Digital Elevation Model (DEM) is an essential aspect in the remote sensing (RS) domain to analyze various applications related to surface elevations. Here, we address the generation of high-resolution (HR) DEMs using HR multi-spectral (MX) satellite imagery as a guide by introducing a novel hybrid transformer model consisting of Densely connected Multi-Residual Block (DMRB) and multi-headed Frequency Selective Graph Attention (M-FSGA). To promptly regulate this process, we utilize the notion of discriminator spatial maps as the conditional attention to the MX guide. Further, we present a novel adversarial objective related to optimizing Sinkhorn distance with classical GAN. In this regard, we provide both theoretical and empirical substantiation of better performance in terms of vanishing gradient issues and numerical convergence. Based on our experiments on 4 different DEM datasets, we demonstrate both qualitative and quantitative comparisons with available baseline methods and show that the performance of our proposed model is superior to others with sharper details and minimal errors.
\end{abstract}

\section{Introduction}
\label{sec:intro}
\setlength{\abovedisplayskip}{0pt}
\setlength{\belowdisplayskip}{0pt}
\vspace{-0.3cm}
The Digital Elevation Model (DEM) is a digital representation of any three-dimensional surface. It is immensely useful in precision satellite data processing, geographic information systems, hydrological studies, urban planning \cite{50}, and many other key applications. The main sources of DEM generation are terrestrial, airborne, or spaceborne, depending on the platform used for data acquisition. However, each of these scenarios has its own set of advantages and disadvantages. While elevation models generated using terrestrial and airborne systems have a high spatial resolution, their coverage is quite restricted and they typically suffer from several issues and systematic errors \cite{2}. Space-borne missions such as SRTM, and ASTER \cite{503,504}, on the other hand, have almost global coverage but lack the spatial resolution. Due to the emerging significance and diverse applications of DEM, both its accuracy and resolution have a substantial impact in different fields of operation \cite{4}. However, HR DEM products are expensive, as they require special acquisition and processing techniques. As an alternative to generating HR DEM from scratch, enhancing the resolution (super-resolution) of existing DEM datasets can be seen as the most optimal strategy to address the shortfall. Hence, we intend to take a step in this direction to generate HR DEM and, to make it more tractable, we formulate this problem in an image super-resolution (SR) setting. As shown in Figure \ref{fig:fig1_results}, our primary objective is to synthesize HR DEM provided its coarser resolution and existing False Colour Composite (FCC) of HR MX imagery.\par
\begin{figure}[tb]
    \centering
    \includegraphics[width=\textwidth]{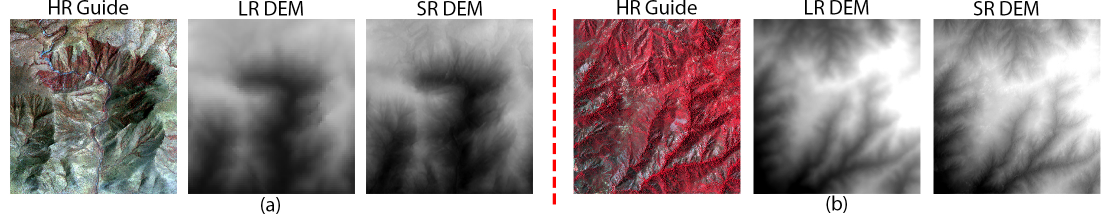}
    \caption{Two sample results of DEM SR consisting HR FCC of NIR(R), R(G), and G(B), Bicubic interpolated LR DEM, and Generated HR DEM, respectively.}
    \label{fig:fig1_results}
\end{figure}
Recent advances in deep learning (DL) show compelling progress over conventional approaches for various computer vision applications like image or video SR. However, we found that very few methods approach the problem of DEM SR, especially, for real-world datasets. We propose a novel framework, which effectively addresses this problem. Our key contributions can be summarized as
\begin{enumerate}
    \item We propose a novel architecture for DEM SR based on a hybrid transformer block consisting of a Densely connected Multi-Residual Block (DMRB) and multi-headed Frequency Selective Graph Attention (M-FSGA),  which effectively utilizes information from an HR MX image as a guide by conditioning it with a discriminative spatial self-attention (DSA).
    \item We develop and demonstrate SiRAN, a framework based on Sinkhorn regularized adversarial learning. We provide theoretical and empirical justification for its effectiveness in resolving the vanishing gradient issue while leveraging tighter iteration complexity.
    \item We generate our own dataset where we take realistic coarse resolution data instead of considering bicubic downsampled HR image as input.
    \item We perform experiments to assess the performance of our model along with ablation studies to show the impact of the different configuration choices.
\end{enumerate} 
\section{Related Work}
\label{sec:formatting}
Traditional DEM super-resolution (SR) methods include interpolation-based techniques like linear, and bicubic, but they under-perform at high-frequency regions producing smoothed outputs. To preserve edge information, multiple reconstruction-based methods like steering kernel regression (SKR) \cite{9} or non-local means (NLM) \cite{10}, have also been proposed. Though they can fulfill their primary objective, they cannot produce SR DEM at a large magnification factor. \par
DEM is an essential component for RS applications, but research on DEM SR is still limited. After the introduction of SR using Convolutional Neural Network (SRCNN) in the category of single image SR (SISR), its variant D-SRCNN was proposed by \cite{20} to address the DEM SR problem. Later, Xu \textit{et al.} \cite{21} uses the concept of transfer learning where an EDSR (Enhanced Deep SR) \cite{22}, pre-trained over natural images, is taken to obtain an HR gradient map which is fine-tuned to generate HR DEM. After the introduction of Generative Adversarial Network (GAN), a substantial number of methods have evolved in the field of SR like Super-resolution using GANs (SRGAN). Based on this recently, Benkir \textit{et al.} \cite{23} proposed a DEM SR model, namely D-SRGAN, and later they suggested another model based on EffecientNetV2 \cite{24} for DEM SISR.  {Although D-SRGAN produces good perceptual SR DEMs, it usually results in noisy predicted samples.} They also suffer from issues of conventional GAN, mode collapse, and vanishing gradients. To resolve this, Wasserstein GAN (WGAN) \cite{31} and its other variants \cite{500} have been introduced. However, these methods are computationally expensive, which can be untangled by introducing an entropic regularization term \cite{33}. In this study, we explore the efficacy of sinkhorn distance \cite{36} in DEM SR, which is one of the forms of entropic optimal transport (EOT).\par 

 {Recently, Li \textit{et al.} \cite{gisr,gisr2} proposed DEM SR algorithms using a global Kriging interpolation based information supplement module and a CNN based local feature generation module. It results preferably as a SISR technique, but, in practical scenarios, it generates artifacts near boundary regions and are unable to reproduce the very fine ground truth (GT) details in the predicted SR. Hence, here we propose a} guided SR technique which is a key research area in computer vision, especially for depth estimation. One of the pioneering works in this domain is \cite{506}, where Kim \textit{et al.} proposes Deformable Kernel Networks (DKN) and Faster DKN (FDKN) which learn sparse and spatially invariant filter kernels. Later, He \textit{et al.} \cite{507} exerts a high-frequency guided module to embed the guide details in the depth map. Recently, Metzger \textit{et al.} \cite{508} has achieved baseline performance by adapting the concept of guided anisotropic diffusion with CNNs. Our proposed method aligns with such depth SR methods as we leverage important HR MX features to generate SR DEM. To address this promptly, we incorporate a graph-based attention due to their efficacy in representation learning for image restoration tasks \cite{graph4,graph5}. However, these works are extended versions of graph neural networks (GNNs) which suffer from over-smoothing problems. To resolve this, \cite{graph7,graph6} utilizes GNN based on filtering in the frequency domain. Despite its efficacy in different DL tasks, it is not properly explored for vision tasks. Hence, here we design our graph attention based on its selected frequencies.

\begin{figure}[tb]
    \centering
    \includegraphics[width=0.9\textwidth]{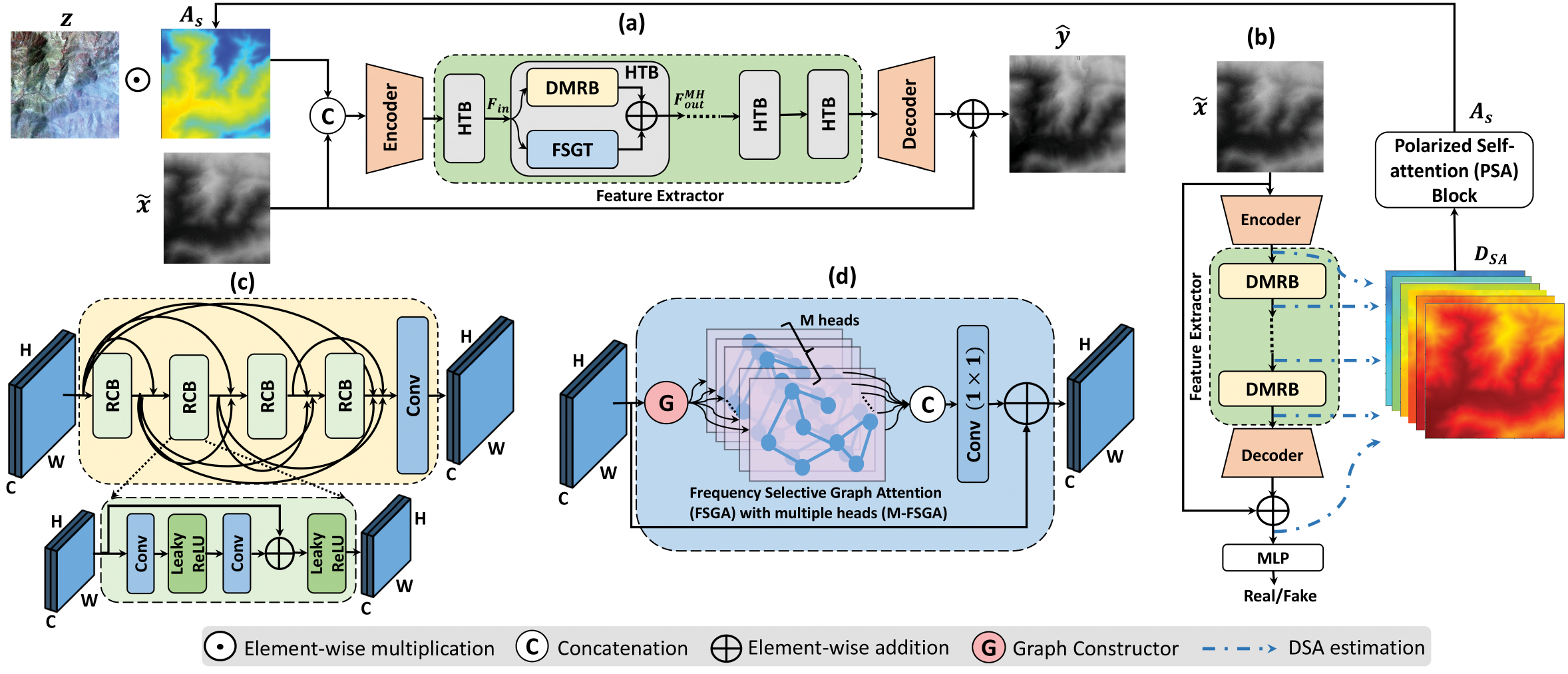}
    \caption{Overview of proposed framework. (a) The generator $G$ have multiple HTBs with parallelly connected (c) DMRB and (d) FSGT. Given guide $\mathbf{z}$ and upsampled LR DEM $\mathbf{\tilde{x}}$ to $G$, each HTB extracts global selective frequency information by FSGT and dense local features via DMRBs in latent space. (b) The discriminator $D$ consists of only DMRBs. Besides classifying predicted $\mathbf{\hat{y}}$ and GT $\mathbf{y}$ as real or fake, $D$ also estimates DSA $\mathbf{D_{SA}}$ with input $\mathbf{\tilde{x}}$. $\mathbf{D_{SA}}$ is passed through a PSA \cite{29} block to estimate $\mathbf{A_s}$ which acted as spatial attention for HR guide $\mathbf{z}$ during passing it to $G$ along $\mathbf{\tilde{x}}$.}
    \label{fig:fig2}
\end{figure}
\section{Methodology}
In Figure \ref{fig:fig2}, we have illustrated the architecture of our framework. The generator $G$ takes upsampled low-resolution (LR) DEM $\mathbf{\tilde{x}}$, and HR MX image guide $\mathbf{z}$, consisting FCC of NIR, red and green bands as input. Let $\mathbf{z} \sim \mathbb{P}_{Z}$, where $\mathbf{z} \in \mathbb{R}^{H \times W \times 3}$ with $\mathbb{P}_{Z}$ being the joint distribution of FCC composition and $ \mathbf{\tilde{x}}\sim \mathbb{P}_{\tilde{x}}$, where $\mathbb{P}_{\tilde{x}}$ constitute of upsampled LR DEM with $\mathbf{\tilde{x}} \in \mathbb{R}^{H \times W}$. Let $\mathbf{\hat{y}} \sim \mathbb{P}_{G_\theta}$ {be the predicted SR DEM} where $\mathbb{P}_{G_\theta}$ is the generator distribution parameterized by $\theta \in \Theta$, parameters of set of all possible generators. Let $\mathbf{y} \sim \mathbb{P}_{y}$ with $\mathbb{P}_{y}$ represents the target HR DEM distribution. The discriminator $D$ classifies $\mathbf{y}$ and $\mathbf{\hat{y}}$ as real or fake, and is assumed to be parameterized by $\psi \in \Psi$, parameters of a set of all possible discriminators. Our $D$ is also designed to estimate spatial attention $\mathbf{D_{SA}}$ from its latent space features with LR DEM $\tilde{x}$ as input as shown in Figure \ref{fig:fig2}. Since $\mathbf{D_{SA}}$ contains discriminative information of HR DEM, it acts as spatial attention for $\mathbf{z}$ allowing the model to focus on salient parts of it and avoid generating out-of-distribution (OOD) image information in the predicted SR DEM. To ensure this further, we process $\mathbf{D_{SA}}$ through a self-attention (SA) block PSA \cite{29} to remove redundant semantics, resulting in an enhanced representative attention map $\mathbf{A_s}$ as demonstrated in Figure \ref{fig:fig2}. Therefore, the predicted SR DEM ($\mathbf{\hat{y}}$) is estimated as $\mathbf{\hat{y}} = G(\mathbf{\tilde{x}},\mathbf{z}\odot \mathbf{A_s})$, where $\odot$ denotes element-wise multiplication.
\vspace{-0.2cm}
\subsection{Network Architecture}
As shown in Figure \ref{fig:fig2}, $G$ is designed based on a novel hybrid transformer block (HTB) {\cite{htb1,htb2} due to their effectiveness in capturing both long-distance as well as local relations in image restoration tasks}. Our HTB consists of a DMRB and a FSGT block. DMRB is developed based on ResNet and DenseNet by using both skip and dense connections. Each DMRB block is constituted of multiple densely connected Residual Convolution Blocks (RCBs). DMRB enables efficient context propagation and also stable gradient flow throughout the network while allowing local dense feature extraction. We introduce FSGT to leverage the extraction of global structural and positional relationships between spatially distant but semantically related regions. We use similar design for $D$. Both incorporate an encoder followed by a feature extractor and finally, a decoder. The feature extractor in $G$ consists of six HTBs while for $D$, it only consists of six DMRBs to extract dense discriminative latent space features, which are used as spatial attention to the HR MX guide. $D$ also adds a Multi-Layer Perceptron (MLP) layer to map its latent features into the required shape. We avoid using batch normalization as it degrades the performance and gives sub-optimal results for image SR \cite{27} tasks. Next, we discuss the functionality of FSGT and DSA.
\vspace{-.35cm}
\subsection{Frequency Selective Graph Transformer (FSGT) Module}\label{3.2}
To exploit high-frequency sharp details from HR guide and enhance latent feature representations, we propose a novel graph transformer, FSGT. As shown in Figure \ref{fig:graph}, for a given input $\mathbf{F_{in}} \in \mathbb{R}^{H \times W \times C}$, FSGT extracts $N$ patches using the patch generation method in W-MSA to construct the graph followed by a FSGA block for graph processing. A graph is represented as $\mathcal{G}= (\mathcal{V}, \mathcal{E})$ with nodes $\mathcal{V} = \{v_i | v_i \in \mathbb{R}^{hw \times c}, i = 1,..., N\}$, where $h$, $w$ and $c$ denotes height, width and channels for each patch represented as node and $\mathcal{E}$ is the set of all the edges connecting these nodes. The edge weights are defined by an adjacency matrix $\mathbf{\mathcal{A}}\in \mathbb{R}^{N \times N}$. The value of $N$ is decided by the shape of each patch $(h,w)$. \par
As shown in Figure \ref{fig:graph} (a), we build the graph connections by computing the similarities \cite{graph1} between the nodes after the linear transformation as $\mathcal{A}_{i,j} = \langle f_1(v_i), f_2(v_j) \rangle$, where $\langle \cdot,\cdot \rangle$ is the inner product, $v_i$ and $v_j$ are $i$-th and $j$-th node, and $f_1$ and $f_2$ corresponds to $1\times 1$ convolution. However, the generated graph $\mathcal{G}$ is dense connecting every node to every other node. Thus, low similarities between some nodes confuse the model on how close different nodes are in the graph. This redundant information will hamper the objective and quality of graph reconstruction. To tackle this, we design FSGA to focus on high-frequency features and also generate a sparse representative graph.\par
Figure \ref{fig:graph}(b) shows the detailed workflow of FSGA. Initially, the nodes $\mathcal{V}$ are flattened out and converted to a matrix $\mathbf{X}\in \mathbb{R}^{N\times hwc}$ as shown in Figure \ref{fig:graph}(a). It is later projected to query ($\mathbf{Q}$), key ($\mathbf{K}$) and value ($\mathbf{V}$) matrices with $\mathbf{Q} = \mathbf{X}\mathbf{W}_q$, $\mathbf{K} = \mathbf{X}\mathbf{W}_k$ and $\mathbf{V} = \mathbf{X}\mathbf{W}_v$, with $\mathbf{W}_q$, $\mathbf{W}_k$, and $\mathbf{W}_v$ being learnable projection weights. However, instead of using $\mathbf{K}$ directly, we filter out certain nodes in $\mathbf{X}$ based on graph Fourier transform (GFT) to generate filtered graph matrix as $\mathbf{\Bar{X}}$. From this the updated key matrix is computed as $\mathbf{\hat{K}} = \mathbf{\Bar{X}}\mathbf{W}_k$ which is used to get the attention as $\mathbf{A} = Softmax(\mathbf{Q}\mathbf{\hat{K}}^T)/\sqrt{d})$. \par
\begin{figure}[tb]
    \centering
    \includegraphics[width=.9\textwidth]{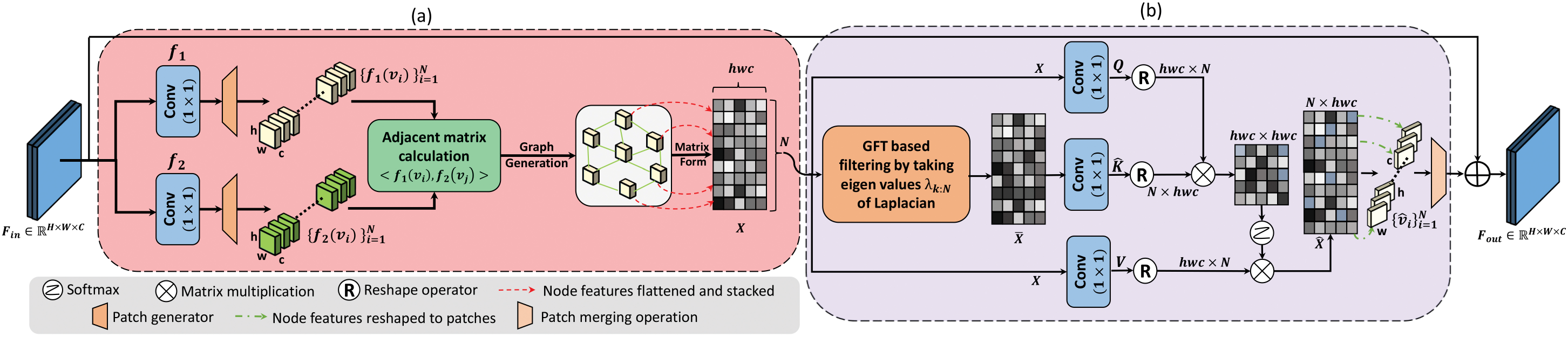}
    \caption{Workflow of FSGT, (a) graph construction mechanism, (b) FSGA block}
    \label{fig:graph}
    \vspace{-0.4cm}
\end{figure}

Graph signals can be analyzed in the frequency domain \cite{graph3} by using normalized Laplacian $\mathbf{\mathcal{L}} = \mathbf{I} - \mathbf{\mathcal{D}}^{-\frac{1}{2}}\mathbf{\mathcal{A}}\mathbf{\mathcal{D}}^{-\frac{1}{2}}$, where $\mathbf{I}$ is the identity matrix and $\mathbf{\mathcal{D}}$ is the diagonal matrix with $\mathbf{\mathcal{D}}_{ii} = \sum_{j}\mathbf{\mathcal{A}}_{ij}$. Taking the eigen-decomposition of $\mathbf{\mathcal{L}}$, we get: $\mathbf{\mathcal{L}} = \mathbf{\mathcal{P}}\mathbf{\Lambda}\mathbf{\mathcal{P}}^{-1}$, where $\mathbf{\mathcal{P}}$ is the eigen-vector matrix and $\mathbf{\Lambda} = diag([\lambda_1, \ldots, \lambda_N])$ is the diagonal eigen-value matrix with eigen values $\lambda_i \, \forall i\in \{1,\ldots, N\}$ ordered in a ascending order. Then, the GFT of $\mathbf{X}$ is defined as $\mathbf{\tilde{X}} = \mathcal{F}_g(\mathbf{X}) = \mathbf{\mathcal{P}}^T\mathbf{X}$, where $\mathbf{\mathcal{P}}\in \mathbb{R}^{hwc\times N}$ (for this section, we use tilde for frequency domain signal). Similarly, the inverse GFT (IGFT) is written as, $\mathbf{X} = \mathcal{F}_g^{-1}(\mathbf{\tilde{X}}) = \mathbf{\mathcal{P}}\mathbf{\tilde{X}}$. $\mathcal{F}_g(\cdot)$ and $\mathcal{F}_g^{-1}(\cdot)$ denotes GFT and IGFT operation. Hence in GFT, the time domain is graph space while the frequency domain is the eigen values $[\lambda_1, \ldots, \lambda_N]$ with each $\lambda_i$ being related to a particular frequency. To estimate the high-frequency, we consider only higher-order eigen values as $\lambda_1<\lambda_2\leq \ldots \ldots \leq \lambda_N$. It results in a sparse graph representation with significant frequency elements by blacking out low-weighted edges as they result in lower eigen values. Hence, we define a vector $\mathbf{\tilde{h}} = \begin{bmatrix}
    \mathbf{0} \: \:
    \mathbf{1}
\end{bmatrix}^T$ to act as a filter in frequency domain, where $\mathbf{0} = \{0\}^{k\times hwc}$ is all-zero matrix, $\mathbf{1} = \{1\}^{(N-k)\times hwc}$ is all-one matrix and $k$ is related to cut-off eigen value $\lambda_k$. The final filtered graph matrix is obtained as equation \ref{eq:graph1}.
\begin{equation}\label{eq:graph1}
    \mathbf{\Bar{X}} = \mathcal{F}_g^{-1} (\mathbf{\tilde{h}} \odot \mathcal{F}_g(\mathbf{X})) = \mathbf{\Bar{\mathcal{P}}}\mathbf{\Bar{\mathcal{P}}}^T\mathbf{X},
\end{equation}
where, $\mathbf{\Bar{\mathcal{P}}} = \mathbf{\mathcal{P}}_{:,k:N}$ are first $k$ eigen vectors. Hence, the node feature aggregation occurs by taking a sparse representative version of $\mathbf{\mathcal{A}}$. It also reduces the computational complexity of our attention module. As we are blacking out $k$ insignificant patches during key estimation, the effective complexity of our overall attention module is $\mathcal{O}((N-k)hwc)$ while it is $\mathcal{O}(Nh^2w^2c)$ for regular MSA. \par 
Using $\mathbf{\Bar{X}}$, we estimate the attention weights as $\mathbf{\hat{X}}$ as shown in Figure \ref{fig:graph} (b), from which the updated node feature patches are generated as $\hat{\mathcal{V}} = \{\hat{v}_i | \hat{v}_i \in \mathbb{R}^{hw \times c}\}$ by reshaping each node $\hat{v}_i$. The output of a FSGA is computed as $\mathbf{F_{out}} = \mathbf{F_{in}} + patch\_merger({\{\hat{v}\}}_{i=1}^N)$. For patch merging, we adapt the method used in W-MSA. We also employ muti-headed attention (M-FSGA) and to stabilize our training process, we dynamically select the value of $k \in \{\lfloor \frac{N}{2} \rfloor, \ldots, N-1\}$ for different heads to ensure not to miss out significant features at different frequencies. The outcomes of M-FSGA $(\{\mathbf{F_{out}^j}\}_{j=1}^{j=M})$ are passed through a Feed Forward Network (FFN) consisting of a concatenate and $1\times 1$ convolution block to aggregate them and project them to a desired shape as shown in Figure \ref{fig:fig2} (d). 

\subsection{Discriminator Spatial Attention (DSA)}\label{3.3}
The feature maps from the latent space of $D$ can be viewed as spatial attention to the HR guide $\mathbf{z}$. Since $D$ performs binary classification, apparently, it captures the discriminative features in latent space. \cite{28} introduced the concept of transferring these domain-specific latent features as attention to $G$. We use this similar notion to help $G$ focus on the salient parts of the HR guide while also helping to avoid the generation of redundant image features in SR DEM.\par
Therefore, besides classification, $D$ has another major functional branch, $D_{SA}$, to approximate spatial attention maps. For any input $\mathbf{m}$, $D_{SA}$ is used to estimate the normalized spatial feature maps, $D_{SA}: \mathbb{R}^{H \times W} \rightarrow[0,1]^{H \times W}$. Let $D$ consist of $t$ DMRBs and $a_i$ be the activation maps after $i^{th}$ DMRB with $c$ channels, such that $a_i \in \mathbb{R}^{H \times W \times c}$. We select $t$ different attention maps after $t$ DMRBs since at different depths, $D$ focuses on different features. Eventually, we calculate these attention coefficients according to \cite{28}, ${D}_{SA}(\mathbf{m}) =\sum_{i=1}^{t} \sum_{j=1}^{c}|a_{i j}(\mathbf{m})|$. \par
 To estimate the attention, we use upsampled LR DEM $\mathbf{\tilde{x}}$ as unlike image-to-image translation in \cite{28}, we do not have HR samples in the target domain during testing. Hence, we use domain adaptation loss from \cite{30} to estimate sharper latent features. The final attention maps $\mathbf{A_s}$ are derived by passing $D_{SA}$ through a PSA \cite{29} to exclude redundant features while highlighting key areas. It is chosen because of its ability to retain a high internal resolution compared to other SA modules. Next, we discuss the theoretical framework for optimizing our model.  
 \subsection{Theoretical framework}
 We train our model with SiRAN, a novel framework regularizing traditional GAN with Sinkhorn distance. Compared to WGAN and its variants which are designed to solve the Kantarovich formulation of OT problems to minimize the Wasserstein distance, SiRAN  showcases favourable sample complexity of $\mathcal{O}(n^{-1/2})$ \cite{34} (for WGAN, it is $\mathcal{O}(n^{-2/d})$ \cite{32}), given a sample size $n$ with a dimension $d$. This is because Sinkhorn is estimated based on entropic regularization. Another key issue with WGANs is the vanishing gradient problem near the optimal point resulting in a suboptimal solution. SiRAN avoids such scenarios, as it provides better convergence and tighter iteration complexity as we derive later.\par
 Let $\mu_{\theta} \in \mathbb{P}_{G_\theta}$ and $\nu \in \mathbb{P}_{y}$ be the measure of generated and true distribution with support included in a compact bounded set $\mathcal{X}, \mathcal{Y} \subset \mathbb{R}^{d}$, respectively. Therefore, the EOT \cite{505} between the said measures can be defined using Kantarovich formulation as shown in equation \ref{eq14} where we assume $\mathbf{\hat{y}} = {G}(\mathbf{\tilde{x}}, \mathbf{z}\odot A_s(\mathbf{\tilde{x}}))$.
\begin{equation}\label{eq14}
\mathcal{W}_{C,\varepsilon}(\mu_{\theta},\nu) = 
\inf_{\pi \in \Pi(\mu_{\theta},\nu)} \mathbb{E}_{\pi}[C(\mathbf{\hat{y}},\mathbf{y})] + \varepsilon I_{\pi}(\mathbf{\hat{y}},\mathbf{y}),
 \, I_{\pi}(\mathbf{\hat{y}},\mathbf{y})) = \mathbb{E}_{\pi}[\log(\frac{\pi(\mathbf{\hat{y}},\mathbf{y}))}{\mu_{\theta}(\mathbf{\hat{y}})\nu(\mathbf{y})}],
\end{equation}
where, $\Pi(\mu_{\theta},\nu)$ is the set of all joint distribution on $\mathcal{X} \times \mathcal{Y}$ with marginals $\mu_{\theta}$ and $\nu$, $C: \mathcal{X} \times \mathcal{Y} \rightarrow \mathbb{R}$ is the cost of transferring unit mass between locations $\mathbf{\hat{y}} \in \mathcal{X}$ and $\mathbf{y} \in \mathcal{Y}$, and the regularization $I_{\pi}(\cdot)$ is the mutual information between two measures \cite{35} with $\varepsilon$ as its weight. When $C(\cdot)$ is distance-metric the solution of equation \ref{eq14} is referred to entropic Wasserstein distance between two probability measures. To fit $\mu_\theta$ to $\nu$, $\mathcal{W}_{C,\varepsilon}(\mu_{\theta},\nu)$ is to be minimized which can be treated as loss function for $G$ \cite{31}. However, it has one major issue of being strictly larger than zero, i.e. $\mathcal{W}_{C,\varepsilon}(\nu,\nu) \neq 0$ which is resolved by adding normalizing terms to equation \ref{eq14} leading to the Sinkhorn loss \cite{36} as defined below.
\begin{equation}\label{eq17}
\begin{aligned}
\mathcal{S}_{C,\varepsilon} = \mathcal{W}_{C,\varepsilon}(\mu_\theta, \nu) - \frac{1}{2} \mathcal{W}_{C,\varepsilon}(\mu_\theta, \mu_\theta) - \frac{1}{2} \mathcal{W}_{C,\varepsilon}(\nu, \nu).
\end{aligned}
\end{equation}
Based on the value of $\varepsilon$, equation \ref{eq17} shows asymptotic behaviour \cite{36}. When $\varepsilon\rightarrow0$, it recovers the conventional OT problem, while $\varepsilon\rightarrow \infty$, it converges to maximum mean discrepancy (MMD). Therefore, the Sinkhorn loss interpolates between OT loss and MMD loss as $\varepsilon$ varies from $0$ to $\infty$ leveraging the concurrent advantage of non-flat geometric properties of OT loss and, high dimensional rigidity and energy distance properties of MMD loss (when $C=||\cdot||_p$ with $1<p<2$). Apart from this, the selection of $\varepsilon$ also affects the overall gradients of $G$, which eventually results in preventing vanishing gradient problems near the optimal point. This can be established from the smoothness property of $\mathcal{S}_{C,\varepsilon}(\mu_\theta, \nu)$ with respect to $\theta$. In this context, we propose \textbf{Theorem 1}, where we derive a formulation to estimate the smoothness of Sinkhorn loss.
\begin{thm}[Smoothness of Sinkhorn loss]
\textit{Consider $\mathcal{S}_{C,\varepsilon}(\mu_\theta, \nu)$ be the Sin-khorn loss between measures $\mu_\theta$ and $\nu$ on $\mathcal{X}$ and $\mathcal{Y}$, two bounded subsets of $\mathbb{R}^{d}$, with a $\mathcal{C}^{\infty}$, $L_0$-Lipschitz, and $L_1$-smooth cost function $C$. Then, for $(\theta_1, \theta_2)\in \Theta$,}
\begin{equation}\label{eq18}
\begin{aligned}
\mathbb{E}||\nabla_\theta \mathcal{S}_{C,\varepsilon}(\mu_{\theta_1}, \nu) - \nabla_\theta \mathcal{S}_{C,\varepsilon}(\mu_{\theta_2}, \nu)||
= \mathcal{O}(L(L_1 + \frac{2L_0^2L}{\varepsilon(1+B e^{\frac{\kappa}{\varepsilon}})})) ||\theta_1 - \theta_2||, 
\end{aligned}
\end{equation}
\textit{where $L$ is the Lipschitz in $\theta$, $\kappa = 2(L_0 |\mathcal{X}| + ||C||_{\infty})$}, $B = d.\max (||m||, ||M||)$ with $m$ amd $M$ being the minimum and maximum in set $\mathcal{X}$. Let $\Gamma_{\varepsilon}$ be the smoothness mentioned above, then we get the following asymptotic behavior in $\varepsilon$: 

\quad \quad 1. \textit{as $\varepsilon \rightarrow 0$, $\Gamma_{\varepsilon} \rightarrow \mathcal{O}(\frac{2 L_0^2 L^2}{B \varepsilon e^{\frac{\kappa}{\varepsilon}}})$,} \quad \textit{and,} \quad 2. \textit{as $\varepsilon \rightarrow \infty$, $\Gamma_{\varepsilon} \rightarrow \mathcal{O}(L L_1)$}.
\vspace{-0.25cm}
\end{thm}
\begin{proof}
    Refer to Appendix B in supplementary (supp.).
    \vspace{-0.2cm}
\end{proof}
\textbf{Theorem 1} shows the variation of smoothness of $\mathcal{S}_{C,\varepsilon}(\mu_\theta, \nu)$ with respect to $\varepsilon$. Using this, we can estimate the upper bound of the overall expected gradient of our proposed adversarial set-up. Hence, to formulate this upper bound, we present \textbf{Proposition 1}. Here, we assume $\mathbf{x}=concat(\mathbf{\tilde{x}}, \mathbf{z}\odot A_s(\mathbf{\tilde{x}}))$.
\vspace{-0.2cm}
\begin{proposition}
\textit{Let $l(\cdot)$, $g(\cdot)$ and $\mathcal{S}_{C,\varepsilon}(\cdot)$ be the objective functions related to supervised losses, adversarial loss and Sinkhorn loss with smoothness $\Gamma_{\varepsilon}$, and $\theta^{*}$ and $\psi^{*}$ be the parameters of optimal $G$ and $D$. Let us suppose $l(\mathbf{\hat{y}},\mathbf{y})$, where $\mathbf{\hat{y}}=G_{\theta}(\mathbf{x})$ is $\beta$-smooth in $\mathbf{\hat{y}}$ for some input $\mathbf{x}$. If $||\theta - \theta^{*}||\leq \epsilon$ and $||\psi - \psi^{*}||\leq \delta$, then $||\nabla_\theta \mathbb{E}_{(x,y) \sim \mathcal{X} \times \mathcal{Y}}[l(\mathbf{\hat{y}},\mathbf{y}) + \mathcal{S}_{C,\varepsilon}(\mu_\theta(\mathbf{\hat{y}}), \nu(\mathbf{y})) - g(\psi; \mathbf{\hat{y}})]||\leq L^2 \epsilon(\beta + \Gamma_{\varepsilon}) + L\delta$.}
\end{proposition}
\begin{proof}
    Refer to Appendix C in supp.
    \vspace{-0.2cm}
\end{proof}
In GAN setups as mentioned in \cite{501}, $\epsilon \rightarrow 0$ leads to a vanishing gradient near the optimal region due to reductions in $\delta$. However, regularizing with Sinkhorn introduces an upper bound dependent on $\Gamma_{\varepsilon}$, which varies exponentially with $\varepsilon$ (see \textbf{Proposition 1}). Choosing an appropriate $\varepsilon$ mitigates the vanishing gradient and enhances performance. Additionally, Sinkhorn regularization improves iteration complexity \cite{501}, resulting in faster convergence as established in \textbf{Proposition 2}..
\vspace{-0.2cm}
\begin{proposition}
Suppose the supervised loss $l(\theta)$ is lower bounded by $l^{*}>\infty$ and it is twice differentiable. For some arbitrarily small $\zeta>0$, $\eta>0$ and $\epsilon_1>0$, let $||\nabla g(\psi; \mathbf{\hat{y}})||\geq \zeta$, $||\nabla\mathcal{S}_{C,\varepsilon}(\mu_\theta, \nu)||\geq \eta$ and $||\nabla l(\mathbf{\hat{y}},\mathbf{y})||\geq \epsilon_1$, with $\delta\leq \frac{\sqrt{2\epsilon_1 \zeta}}{L}$, and $\Gamma_{\varepsilon}<\frac{\sqrt{2\epsilon_1 \eta}}{L^2\epsilon}$, then the iteration complexity in Sinkhorn regularization is upper bounded by $\mathcal{O}(\frac{(l(\theta_0) - l^*)\beta_1}{\epsilon_1^2+2\epsilon_1(\zeta+\eta) - L^2(\delta^2+L^2\Gamma_{\varepsilon}^2\epsilon^2)})$, assuming $||\nabla^2l(\theta)||\leq \beta_1$.
\vspace{-0.2cm}
\end{proposition}
\begin{proof}
    Refer to Appendix D in supp.
    \vspace{-0.2cm}
\end{proof}
\begin{corollary}
    Using first order Taylor series, the upper bound in \textbf{Proposition 2} becomes $\mathcal{O}(\frac{l(\theta_0) - l^*}{\epsilon_1^2+\epsilon_1(\zeta+\eta)})$.
    \vspace{-0.2cm}
\end{corollary}
\begin{proof}
    Refer to Appendix D.1 in supp.
\end{proof}
When $\Gamma_{\varepsilon}<\frac{\sqrt{2\epsilon_1 \eta}}{L^2\epsilon}$, the denominator of the derived upper bound in \textbf{Proposition 2} is greater than the same in Theorem 3 of \cite{501}. This is true for almost all valid $\varepsilon$ as we experimentally verify in Appendix E in supp. Therefore, SiRAN has tighter iteration complexity compared to the regular GAN set-ups. \textbf{Corollary 1} also verifies this using a simpler setup, as it increases the convergence rate from $\mathcal{O}((\epsilon_1^2 + \epsilon_1 \zeta)^{-1})$ \cite{501} to $\mathcal{O}((\epsilon_1^2 + \epsilon_1(\zeta+\eta))^{-1})$. Due to these advantages, we regularize the generator loss with Sinkhorn distance as defined below,
\begin{equation}
\mathscr{L}_{OT}= \mathbb{E}_{\mathbf{\tilde{x}} \sim \mathbb{P}_{\tilde{x}}, \mathbf{z} \sim \mathbb{P}_{Z}, \mathbf{y} \sim \mathbb{P}_{y}}\mathcal{S}_{C,\varepsilon}(\mu(\mathbf{\hat{y}}), \nu(\mathbf{y})),
\end{equation}
where $\mu$ and $\nu$ is the measure of generated and true distributions. $\mathscr{L}_{OT}$ is estimated according to \cite{36} which utilizes  $\varepsilon$ and the Sinkhorn iterations $T$ as the major parameters. As Sinkhorn loss also minimizes the Wasserstein distance, it serves the purpose of WGAN to resolve the issues of the original GAN more effectively. Hence, we use original GAN objective function $(\mathscr{L}_{ADV})$ while regularized with Sinkhorn loss. We also regularize the objective function of $G$ with pixel loss $(\mathscr{L}_{P})$ and SSIM loss $(\mathscr{L}_{SSIM})$ to generate samples close to GT in terms of minimizing the pixel-wise differences while preserving the perceptual quality and structural information. Therefore, the overall generator loss is defined as
\begin{equation}
\lambda_P\mathscr{L}_{P} + \lambda_{SSIM}\mathscr{L}_{SSIM} + \lambda_{ADV}\mathscr{L}_{ADV} + \lambda_{OT}\mathscr{L}_{OT},
\end{equation}
where $\lambda_P$, $\lambda_{SSIM}$, $\lambda_{ADV}$ and $\lambda_{OT}$ represent the weight assigned to pixel loss, SSIM loss, adversarial loss, and Sinkhorn loss respectively.\par
Similarly, the objective function of $D$ is designed based on the original GAN. In addition, we include domain adaptation loss \cite{30} $(\mathscr{L}_{DA})$  to enforce the $D$ to mimic the latent features of the HR DEM and sharpen spatial attention maps provided an upsampled LR DEM data. The final objective function of ${D}$ becomes
\begin{equation}\label{eq7}
\min _D - \mathbb{E}_{y \sim \mathbb{P}_y}[\log ({D}(\mathbf{y})))] - \mathbb{E}_{\hat{y} \sim \mathbb{P}_{{G}_\theta}}[\log(1-{D}(\mathbf{\hat{y}}))]+  \lambda_{DA}\mathscr{L}_{D A},    
\end{equation}
where $\lambda_{DA}$ is the assigned weight for $\mathscr{L}_{DA}$ in the discriminator objective. The details of $\mathscr{L}_{ADV}$, $\mathscr{L}_{P}$, $\mathscr{L}_{SSIM}$, and $\mathscr{L}_{DA}$ are discussed in Appendix A in supp.
\section{Experiments}
Here, we discuss the necessary experiments and datasets for DEM SR.
\subsection{Datasets}
DEM SR is a relatively unexplored area that suffers from a lack of realistic datasets. Hence, we generate our own DEM SR dataset for this study. From the real-world application point of view, we use real coarse resolution SRTM DEM with a ground sampling distance (GSD) of 30m as input instead of conventional bicubic downsampled while taking Indian HR DEM (GSD=10m) generated from Cartosat-1 stereoscopic satellite as the GT. For the guide, we take the HR MX data (GSD=1.6m) from the Cartosat-2S satellite. The DEMs are upsampled to the resolution of MX images using bicubic interpolation to generate a paired dataset. This helps in increasing the training samples and also assists the model in learning dense HR features from the guide. The dataset consists of 72,000 patches of size $(128,128)$ including various signatures such as vegetation, mountains, and, water regions. We use 40,000 samples for training, 20,000 for cross-validation, and 12,000 for testing, where 10,000 patches belong to the Indian region and the rest outside India. As GT is only available for Indian regions, our model is trained on limited landscape areas. To check its generalization ability, we test our model on data from the Fallbrook region, US, where Cartosat DEM data is unavailable. For these cases, we validate our result based on available 10m DEM data of 3DEP \cite{509}. We further test our trained model by taking other available 30 m DEM like ASTER \cite{504} and AW3D30 \cite{jaxa}. In these cases, we have taken 5000 samples each from different parts of the India for testing.
\begin{table}[tb]
\caption{Quantitative comparison with state-of-the-art methods for both patches of inside and outside India. First and second methods are highlighted in red and green.}
\centering
\label{tab:my-table}
\resizebox{0.85\linewidth}{!}{%
\begin{tabular}{ccccccccc}
\hline \hline
\textbf{Method}       & \multicolumn{2}{c}{\textbf{RMSE (m)}}                                                              & \multicolumn{2}{c}{\textbf{MAE (m)}}                                                               & \multicolumn{2}{c}{\textbf{SSIM(\%)}}                                                              & \multicolumn{2}{c}{\textbf{PSNR}}                                                                  \\ \hline
\textbf{Dataset}      & \multicolumn{1}{c|}{\textbf{Inside}}                       & \textbf{Outside}                      & \multicolumn{1}{c|}{\textbf{Inside}}                       & \textbf{Outside}                      & \multicolumn{1}{c|}{\textbf{Inside}}                       & \textbf{Outside}                      & \multicolumn{1}{c|}{\textbf{Inside}}                       & \textbf{Outside}                      \\ \hline
\textbf{Bicubic}      & \multicolumn{1}{c|}{21.25}                                 & 23.19                                 & \multicolumn{1}{c|}{22.42}                                 & 22.04                                 & \multicolumn{1}{c|}{71.27}                                 & 66.49                                 & \multicolumn{1}{c|}{30.07}                                 & 27.79                                 \\
\textbf{ENetV2} \cite{24}       & \multicolumn{1}{c|}{20.35}                                 & 30.53                                 & \multicolumn{1}{c|}{18.72}                                 & 28.36                                 & \multicolumn{1}{c|}{69.63}                                 & 60.04                                 & \multicolumn{1}{c|}{31.74}                                 & 25.58                                 \\
\textbf{DKN} \cite{506}         & \multicolumn{1}{c|}{12.89}                                 & 21.16                                 & \multicolumn{1}{c|}{11.18}                                 & 19.78                                 & \multicolumn{1}{c|}{73.59}                                 & 68.45                                 & \multicolumn{1}{c|}{32.09}                                 & 28.22                                 \\
\textbf{FDKN} \cite{506}        & \multicolumn{1}{c|}{13.05}                                 & 21.93                                 & \multicolumn{1}{c|}{11.34}                                 & 20.41                                 & \multicolumn{1}{c|}{74.13}                                 & 66.83                                 & \multicolumn{1}{c|}{32.46}                                 & 27.68                                 \\
\textbf{DADA} \cite{508}        & \multicolumn{1}{c|}{37.49}                                 & 40.89                                 & \multicolumn{1}{c|}{32.17}                                 & 37.74                                 & \multicolumn{1}{c|}{73.32}                                 & 69.86                                 & \multicolumn{1}{c|}{27.94}                                 & 26.78                                 \\
\textbf{ESRGAN} \cite{23}      & \multicolumn{1}{c|}{31.33}                                 & {\color[HTML]{32CB00} \textit{20.45}} & \multicolumn{1}{c|}{25.56}                                 & {\color[HTML]{32CB00} \textit{18.34}} & \multicolumn{1}{c|}{{\color[HTML]{32CB00} \textit{82.48}}} & {\color[HTML]{32CB00} \textit{75.67}} & \multicolumn{1}{c|}{29.88}                                 & {\color[HTML]{32CB00} \textit{29.05}} \\
\textbf{FDSR} \cite{507}          & \multicolumn{1}{c|}{{\color[HTML]{32CB00} \textit{12.98}}} & 30.58                                 & \multicolumn{1}{c|}{{\color[HTML]{32CB00} \textit{10.87}}} & 25.28                                 & \multicolumn{1}{c|}{81.49}                                 & 59.81                                 & \multicolumn{1}{c|}{{\color[HTML]{32CB00} \textit{33.77}}} & 25.59                                 \\
\textbf{SiRAN (ours)} & \multicolumn{1}{c|}{{\color[HTML]{CB0000} \textbf{9.28}}}  & {\color[HTML]{CB0000} \textbf{15.74}} & \multicolumn{1}{c|}{{\color[HTML]{CB0000} \textbf{8.51}}}  & {\color[HTML]{CB0000} \textbf{12.25}} & \multicolumn{1}{c|}{{\color[HTML]{CB0000} \textbf{90.59}}} & {\color[HTML]{CB0000} \textbf{83.90}} & \multicolumn{1}{c|}{{\color[HTML]{CB0000} \textbf{35.06}}} & {\color[HTML]{CB0000} \textbf{31.56}} \\ \hline \hline
\end{tabular}%
}
\vspace{-0.5cm}
\end{table}
\vspace{-0.15cm}
\subsection{Implementation Details}
All the experiments are conducted under identical environments. We use $3\times 3$ convolution kernel and leaky ReLU activation except in the last layer where $1\times 1$ kernel is used without any activation. Each DMRB has 64 convolution operations. For FSGT in HTB, we select patch size as $7\times 7$ and the number of heads in the attention block as $M=16$. We use an ADAM optimizer with a fixed learning rate of 0.0001. During adversarial training, we update the critic once every single update in the generator. We set $\lambda_{DA} = 0.1$, $\lambda_P=100$, $\lambda_{str}=1$, $\lambda_{ADV}=1$ and $\lambda_{OT}=0.01$. For estimating $\mathscr{L}_{OT}$, we set $T=10$ and $\varepsilon=0.1$. The entire framework is developed using PyTorch. All the experiments are performed on 2 Nvidia V100 GPUs. We compare our method with traditional bicubic as well as other learning-based state-of-the-art (SOTA) DEM SR methods \cite{24,gisr,gisr2}. For a fair comparison, we also include recent baseline models for image-guided depth SR \cite{23,506,508,507}. All the learning-based methods are trained on our dataset from scratch according to the respective authors' guidelines. Among them, we train \cite{24,gisr,gisr2} without any guide as there is no provision in including an image guide in these methods, whereas, \cite{23,506,508,507} are trained on our dataset in the presence of the guide due to their similar set-up for guided SR. 
\section{Result Analysis}
Here, we analyze both qualitatively and quantitatively, the quality of generated HR DEM by our proposed method.\begin{figure}[t]
    \centering
    \includegraphics[width=0.93\textwidth]{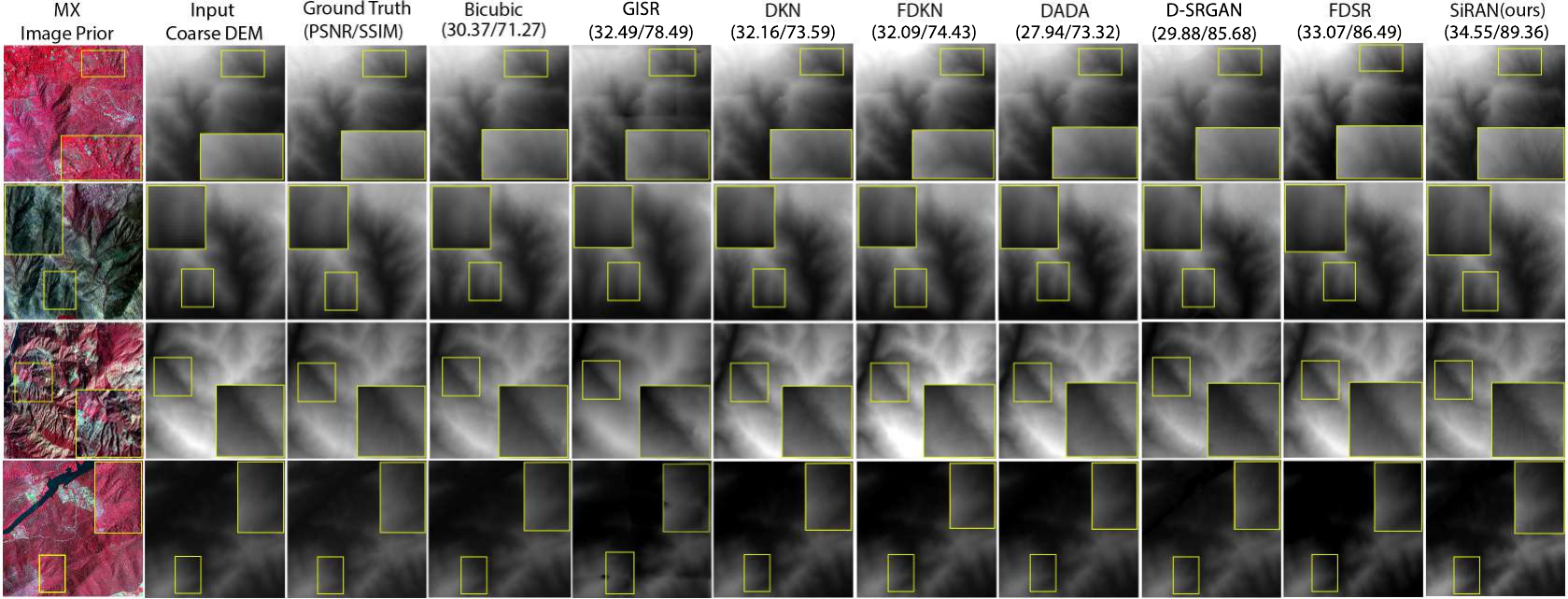}
    \caption{Test results (inside India) for DEM super-resolution (better viewed at 200\%) and comparisons with other baseline methods.}
    \label{fig2}
\end{figure}
\begin{figure}[t]
    \centering
    \includegraphics[width=0.93\textwidth]{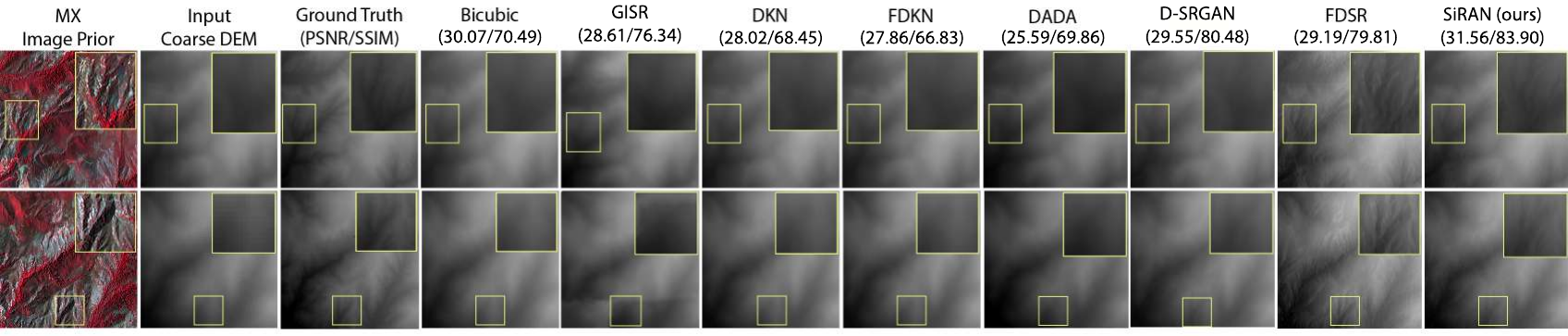}
    \caption{Test results (outside India) for DEM super-resolution (better viewed at 200\%) and comparisons with other baseline methods.}
    \label{fig20}
\end{figure}

\subsection{Quantitative Analysis}
To quantitatively analyze the performance, we use RMSE, MAE, PSNR, and SSIM as the evaluation metrics. Our proposed method outperforms other SOTA methods {for 4 different datasets}, as shown in Table \ref{tab:my-table}. For both inside and outside India images, SiRAN achieves more than 24\% improvement in RMSE and MAE, 8\% in SSIM, and 1.2 dB in PSNR with respect to the second best. Despite having different source domains for reference DEM for outside India cases, SiRAN generates SR DEM closer to GT as depicted in Table \ref{tab:my-table} suggesting better generalization capability of other baseline methods. {This also can be depicted by analyzing on test cases for other LR DEM data like ASTER and AW3D30 as shown in  Table \ref{tab:my-table}. In these cases, SiRAN gains more than 10-18\% improvement in RMSE, 11-27\% in MAE, 4\% in SSIM, and $\sim 1$ dB in PSNR.} Among others, FDSR \cite{507} performs close to our model for Indian patches {as well as for other LR DEM samples}. However, for outside patches, it performs poorly. Although D-SRGAN captures structural details, it has poor RMSE and MAE. Figure \ref{fig3} shows the line profiles of SiRAN and other baselines with respect to GT. Comparatively SiRAN has the lowest bias and follows the true elevation values most closely. This supports the error analysis in Table \ref{tab:my-table}. Table \ref{tab:my-table} shows a comparison of number of parameters and average runtime for $512\times512$ patches. Despite having larger parameters, our model takes comparable inference time due to its effective complexity as discussed in section \ref{3.2}.

\begin{figure}[t]
    \centering
    \includegraphics[width=0.935\textwidth]{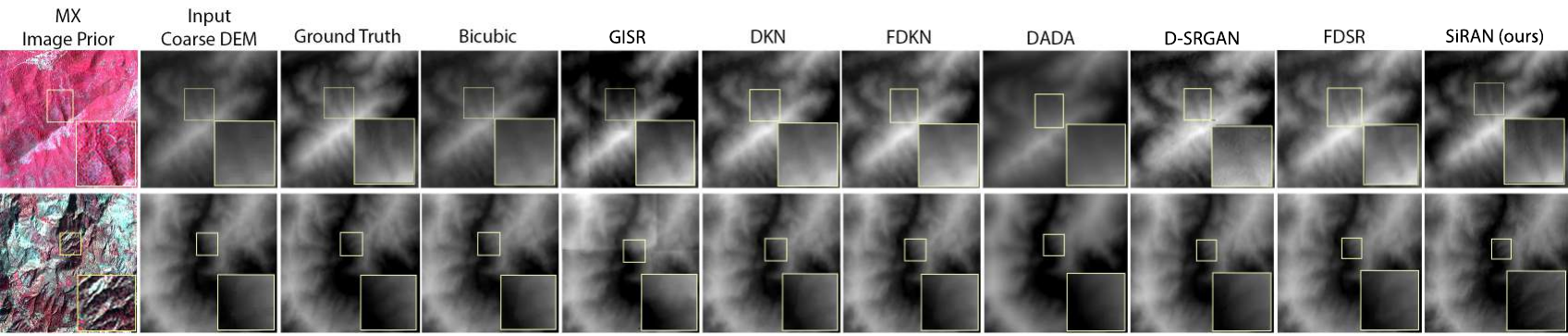}
    \caption{Test results on ASTER (top row) and AW3D30 (bottom row) dataset for DEM super-resolution (better viewed at 200\%) and comparisons with other baseline methods.}
    \label{fig200}
\end{figure}
\begin{figure}[tb]
\begin{minipage}[b]{.45\textwidth}
\centering
\includegraphics[width=0.9\textwidth]{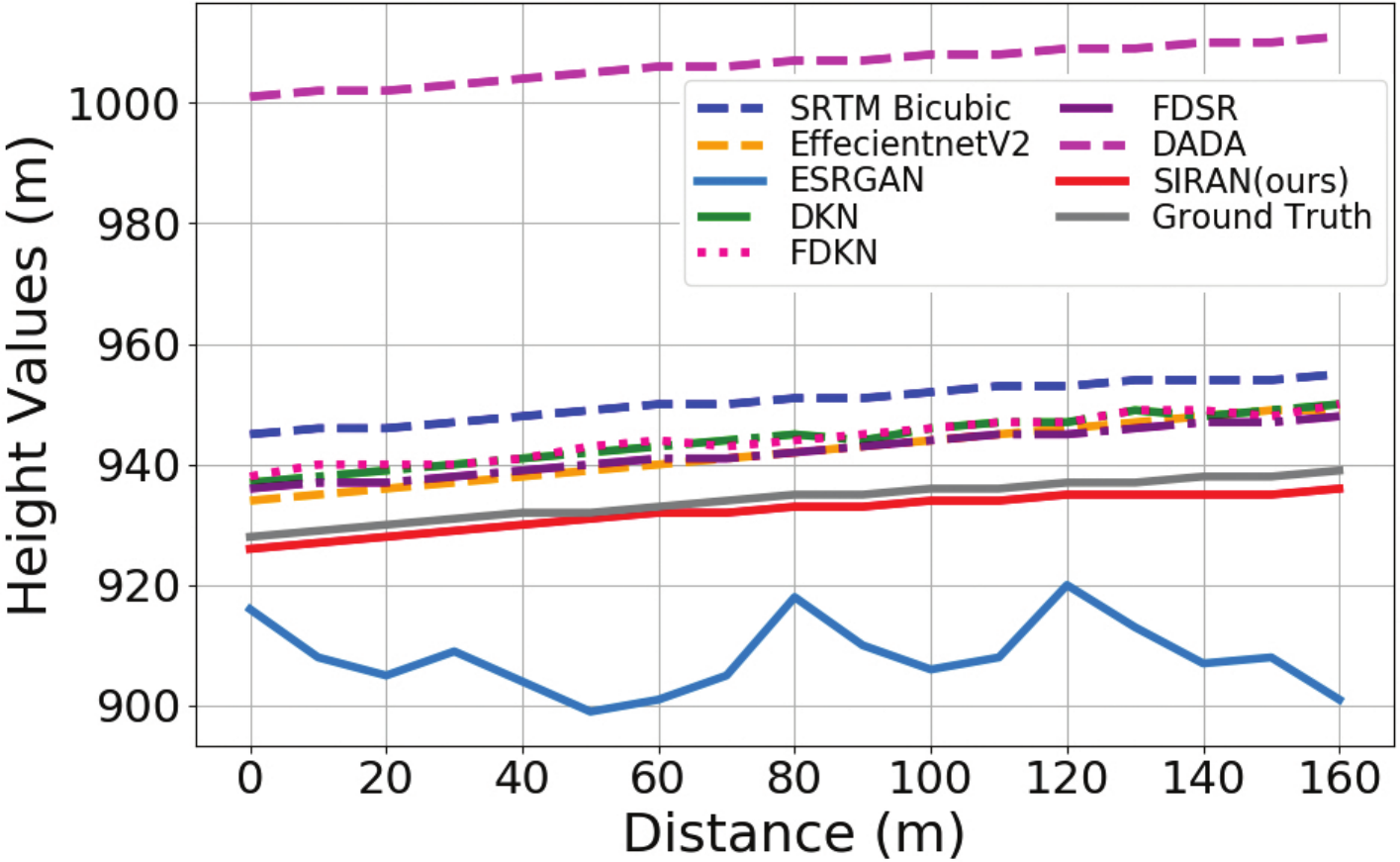}
\caption{Line profile analysis of SiRAN and other baselines.}
\label{fig3}
\end{minipage}
\hfill
\begin{minipage}[b]{.45\textwidth}
\centering
\includegraphics[width=0.9\textwidth]{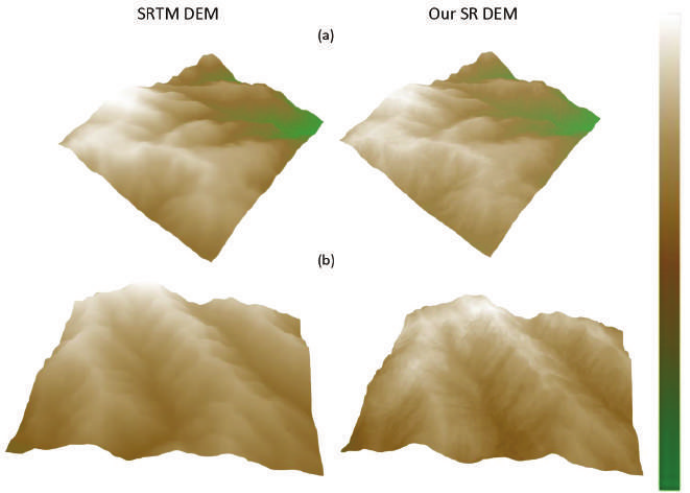}
\caption{Illustration of 3-D visualization of Super-resolved and SRTM DEM}
\label{fig4}
\end{minipage}
\end{figure}
\subsection{Qualitative Analysis}
Figure \ref{fig2} demonstrates the qualitative comparison of DEM SR for patches of India. Clearly, SiRAN highlights key features and comparatively retain more the perceptual quality with respect to GT. D-SRGAN also captures major structural information in its outcomes, however, it tends to produce artifacts and noise in the generated DEM which is depicted in Table \ref{tab:my-table} and Figure \ref{fig3}. In Figure \ref{fig20}, we have compared the outcomes for outside India cases. Here also compared to other SOTA methods, SiRAN is able to generate higher resolution DEM in close proximity to the GT despite having a different source domain. Although FDSR \cite{507} performed well for Indian patches, due to a lack of generalization capability it introduces image details prominently in the generated DEM for test patches outside India. {The generalization ability of these models can also be visualized from \ref{fig200} where we demonstrate visual test cases for LR DEMs of ASTER and AW3D30 datasets. Clearly, SiRAN captures the high-frequency details most effectively in the predicted SR DEM followed by FDSR and D-SRGAN. Among the other models, while DKN and FDKN try to incorporate HR guide details in the SR output, DADA blurs out important features resulting in outputs similar to bicubic interpolation. GISR model also showcases similar results, however, it generates boundary artifacts in their predictions.} In Figure \ref{fig4}, we show 3-D visualization of generated DEMs for a region, where GT is unavailable. We compare it with available SRTM DEM, and clearly, our topographic view of generated DEM captures sharper features in mountainous regions and in the tributaries of the water basin area as shown in Figure \ref{fig4}.
\begin{table}[tb]
\centering
\begin{minipage}[b]{.48\textwidth}
\centering
\caption{Quantitative analysis on effect of different modules for DEM SR.}
\resizebox{0.8\textwidth}{!}{
\begin{tabular}{cccc|cccc}
\hline\hline
\textbf{\begin{tabular}[c]{@{}c@{}}Image\\ Guide\end{tabular}} &
  \textbf{\begin{tabular}[c]{@{}c@{}}DSA \end{tabular}} &
  \textbf{PSA} &
  \textbf{\begin{tabular}[c]{@{}c@{}}FSGT\\ \end{tabular}} &
  \textbf{RMSE (m)} &
  \textbf{MAE (m)} &
  \textbf{SSIM (\%)} &
  \textbf{PSNR} \\ \hline
 \xmark & \xmark & \xmark & \xmark & 20.04         & 17.63         & 75.27          & 30.27          \\
 \cmark & \xmark & \xmark & \xmark & 20.32         & 18.41         & 82.92          & 30.57          \\
 \cmark & \cmark & \xmark & \xmark & 16.06         & 13.62         & 85.68          & 32.08          \\
 \cmark & \cmark & \cmark & \xmark & 13.43         & 11.31         & 87.04          & 32.71          \\
 \cmark & \cmark & \cmark & \cmark & \textbf{9.28} & \textbf{8.51} & \textbf{90.49} & \textbf{35.06} \\ \hline \hline
\end{tabular}}
\label{table2}
\end{minipage}
\hfill
\begin{minipage}[b]{.48\textwidth}
\begin{minipage}[b]{.475\linewidth}
    \caption{Ablation of No. of heads.}
\label{tab:heads}
\resizebox{\textwidth}{!}{%
\begin{tabular}{c|ccc}
\hline
\textbf{\begin{tabular}[c]{@{}c@{}}Number\\ of heads\end{tabular}} & \textbf{\begin{tabular}[c]{@{}c@{}}Params\\ (M)\end{tabular}} & \textbf{PSNR} & \textbf{SSIM} \\ \hline
4                        & 5.29                & 34.34         & 89.04         \\
8                        & 7.41               & 34.55         & 89.36         \\
12                       & 16.37               & 34.59         & 89.64         \\
16                       & 21.36               & 34.61         & 90.09         \\
24                       & 30.22               & 34.72         & 90.13         \\ \hline
\end{tabular}%
}
\end{minipage}
\hfill
\begin{minipage}[b]{.475\linewidth}
\centering
\caption{Model size comparison.}
\label{tab:computation}
\resizebox{\textwidth}{!}{%
\begin{tabular}{c|ccc}
\hline
\textbf{Model} & \textbf{\begin{tabular}[c]{@{}c@{}}Params\\ (M)\end{tabular}} & \textbf{\begin{tabular}[c]{@{}c@{}}FLOPs\\ (G)\end{tabular}}  & \textbf{\begin{tabular}[c]{@{}c@{}}PSNR\\ (dB)\end{tabular}} \\ \hline
SwinIR \cite{swin}        & 11.90               & 215.3                               & 34.41         \\
CAT \cite{cat}            & 16.60               & 360.7                               & 34.16         \\
HAN \cite{han}            & 16.07               & 269.1                               & 33.94         \\
ART \cite{art}            & 11.87               & 278.3                              & 34.25         \\
FSGT (ours)    & 21.36               & 189.4                               & 34.55         \\ \hline
\end{tabular}%
}

\end{minipage}
\end{minipage}
\centering
\vspace{-0.15cm}
\end{table}
\begin{figure}[tb]
\begin{minipage}[b]{.48\textwidth}
\centering
    \includegraphics[width=\textwidth, height=3cm]{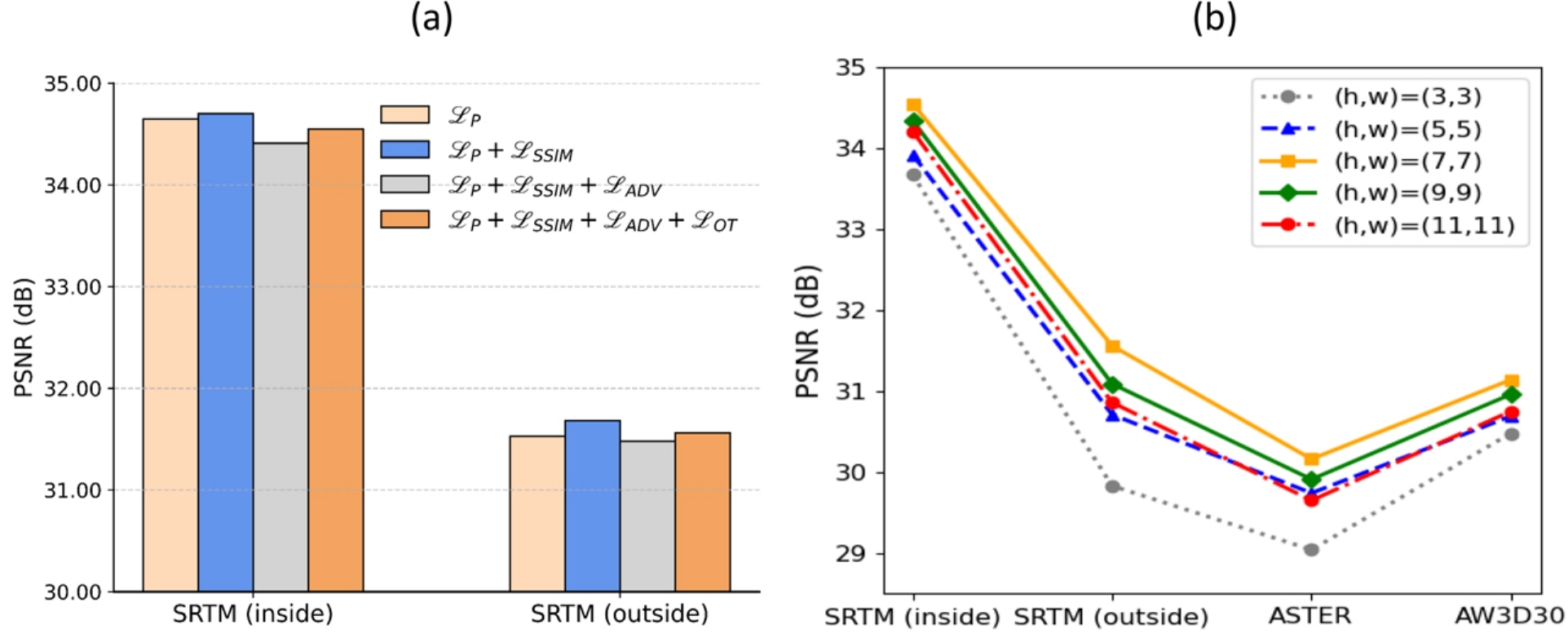}
    \caption{Quantitative ablation study for: (a) introducing different loss functions, and (b) different values of patch size $(h,w)$ on various test dataset.}
    \label{fig:loss_abl4}
\end{minipage}
\hfill
\begin{minipage}[b]{.48\textwidth}
\centering
\includegraphics[width=0.8\textwidth]{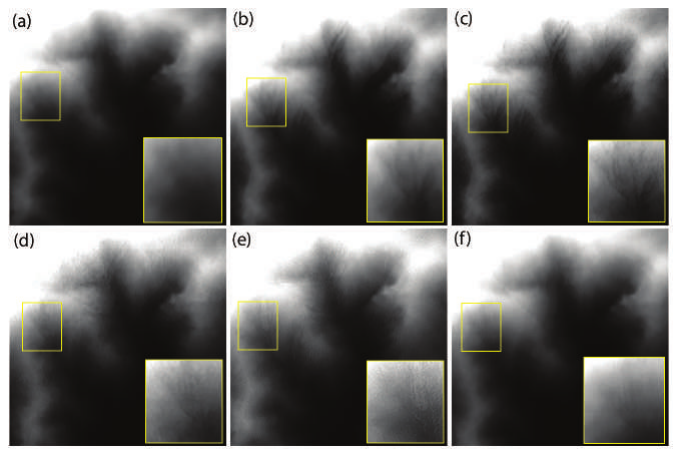}
    \caption{Loss ablation: (a) LR DEM, (b) GT; predicted SR DEM of (c) all losses, (d) $\mathscr{L}_P+\mathscr{L}_{SSIM}+\mathscr{L}_{ADV}$, (e) $\mathscr{L}_P+\mathscr{L}_{SSIM}$, and (f) $\mathscr{L}_P$.}
    \label{fig:loss_abl2}
\end{minipage}
\end{figure}
\subsection{Ablation study}
We discuss different configuration choices we have taken in our designed model for optimal performance in DEM SR in our dataset.\par
\textbf{Choice of different architectural designs:} Table \ref{table2} shows the performance comparisons in terms of different proposed modules. Introducing FSGT brings about the best performance of our framework for DEM SR. However, the utilization of the image guide improves the SSIM only due to its tendency to prominently capture HR MX features in SR DEM. Introducing discriminator spatial attention (DSA) and PSA controls the imitation of guide features phenomenon which results in performance gain in terms of all the metrics. This can also be visualized from Figure \ref{fig5} and \ref{fig6} where we show how $D$ focuses on different features at different depths and also how PSA highlights certain features to give more weight. FSGT further enhances this performance. In this regard, we have also tested with constant $k=\lfloor \frac{3N}{4} \rfloor$, and we have seen more than 0.75 dB performance drop in terms of PSNR and 1.34\% in SSIM.\par
\textbf{Choice of different loss functions:} Figure \ref{fig:loss_abl4} (a) shows the performance of our model with different combinations of loss functions. Introducing $\mathscr{L}_{ADV}$ decreases the PSNR by 0.2-0.3 dB, while adding $\mathscr{L}_{OT}$ improves it by 0.1 dB. Although, it is still less by 0.15 dB compared with $\mathscr{L}_{P} + \mathscr{L}_{SSIM}$ loss combination, the major reason for using $\mathscr{L}_{ADV}$ and $\mathscr{L}_{OT}$ is to improve the overall perceptual quality of SR DEM as shown in Figure \ref{fig:loss_abl2}. However, as depicted in \textbf{Proposition 2}, it provides faster convergence as shown in Figure \ref{fig7}. More experiments are carried out in Appendix E to justify these claims.\par
\begin{figure}[tb]
\begin{minipage}[b]{.32\textwidth}
\centering
\includegraphics[width=0.95\textwidth]{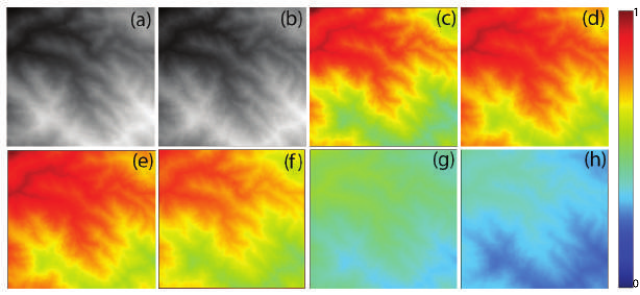}
\caption{(a) Source, (b) Target, (c)-(h) Discriminator spatial attention after each DMRB.}
\label{fig5}
\end{minipage}
\hfill
\begin{minipage}[b]{.32\textwidth}
\centering
\includegraphics[width=\textwidth]{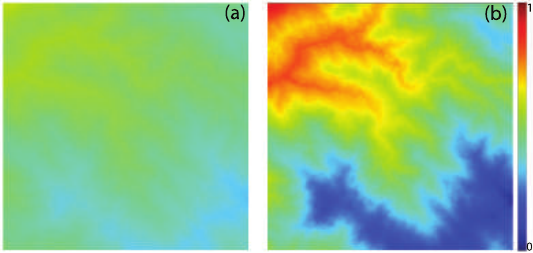}
\caption{Weights of (a) mean DSA ($\textbf{D}_{SA}$),  and (b) after passing it through PSA block.}
\label{fig6}
\end{minipage}
\begin{minipage}[b]{.32\textwidth}
\centering
\includegraphics[width=0.9\textwidth]{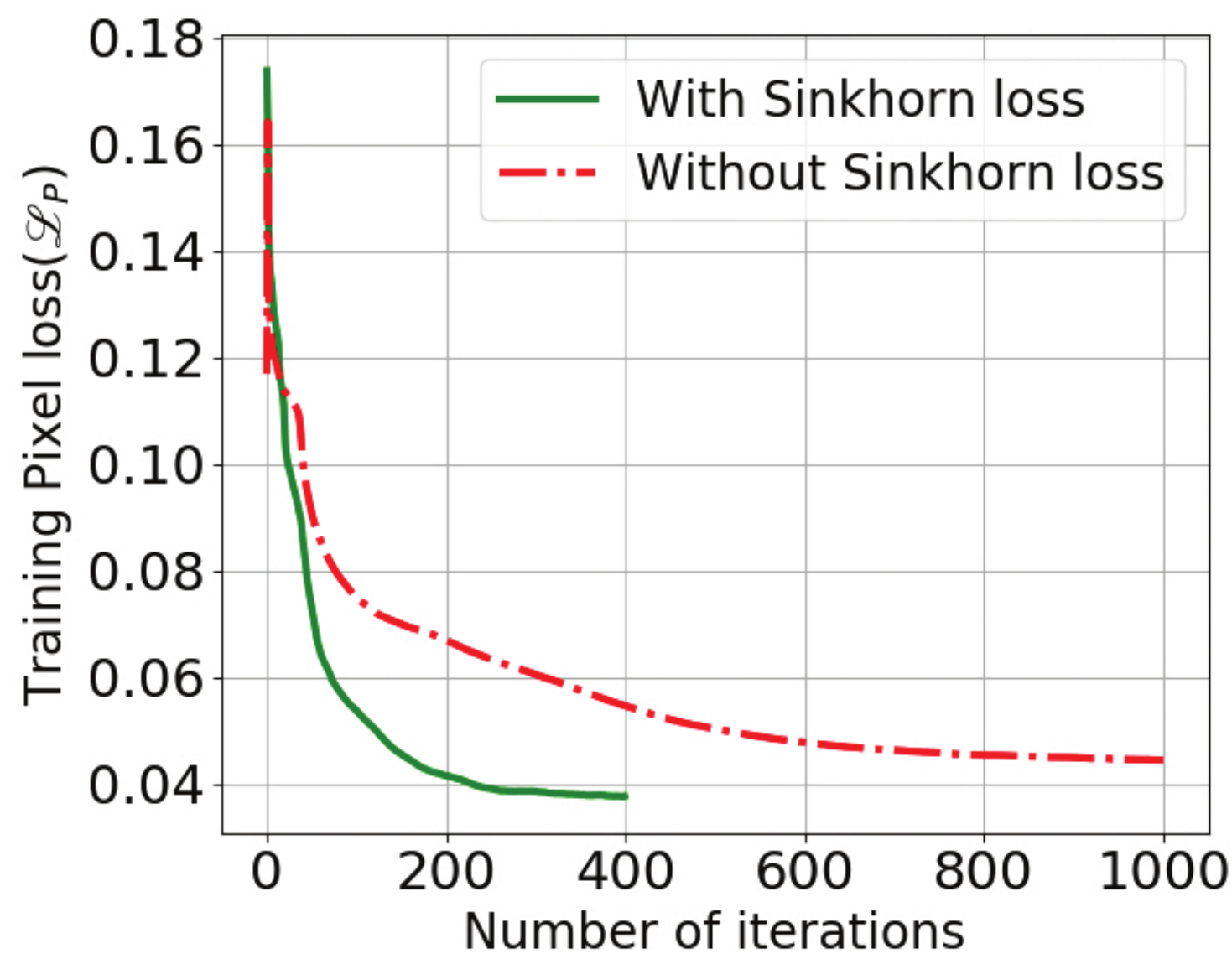}
\caption{Effect of Sinkhorn loss in training convergence.}
\label{fig7}
\end{minipage}
\end{figure}
\textbf{Different patch sizes in FSGT:} Figure \ref{fig:loss_abl4} (b) shows the performance of our model for different patch sizes in FSGT layers. In our case of DEM SR, patch size $7\times 7$ performs the best in terms of PSNR for all of the four datasets.\par
\textbf{Different numbers of heads in M-FSGA:} Table \ref{tab:heads} shows the performance of our model for different numbers of heads ($M$) in proposed M-FSGA. As shown in the table, $M=8$ is the optimal choice in our case. $M=24$ improves the performance by 0.13 dB PSNR but at the cost of 40\% more parameters.\par
\textbf{Model size comparison:} Table \ref{tab:computation} shows the comparison of model size, computational complexity, and performance for DEM SR with respect to popular benchmark transformer models. Clearly, FSGT provides excellent performance while having the least number of FLOPs with competitive model size. 
\section{Conclusion}
In this paper, we present an effective approach for DEM SR using realistic coarse data samples in the presence of an HR MX guide. We propose a novel hybrid transformer model based on FSGT and DMRB. In particular, FSGT is constructed to capture the HR features based on dynamically selected frequencies in a graph attention layer. This also reduces the overall complexity from $\mathcal{O}(Nh^2w^2c)$ to $\mathcal{O}((N-k)hwc)$. To control the in-painting of HR guide features in SR DEM, we also introduce DSA, and through an intense ablation study, we validate the performance of each of these proposed modules. We also present a new adversarial set-up, SiRAN based on Sinkhorn loss optimization. We provided theoretical and empirical evidence to show its efficiency in improving the convergence and speed of training our model. We perform quantitative and qualitative analysis by generating and comparing DEMs related to different signatures {for four different datasets which includes not only the generated inside and outside India test cases corresponding to LR SRTM DEM but also includes LR test samples corresponding to other DEM datasets, ASTER and AW3D30. In all these cases, our model performs preferably by generating close-to-ground truth SR predictions compared to other baseline methods, which showcases its efficiency in capturing high-frequency details as well as better generalization capability.}

\section{Supplementary}
\setcounter{section}{0}
\renewcommand\thesection{\Alph{section}}
\section{Definition of losses used for DEM SR}
In \S 3.3 we have discussed how discriminator spatial attentions are estimated using $D_{SA}(\cdot)$. The domain adaptation loss $\mathscr{L}_{D A}$ is defined as,
\begin{equation}
\mathscr{L}_{D A}=\mathbb{E}_{\mathbf{\tilde{x}} \sim \mathbb{P}_{\tilde{x}}, \mathbf{y} \sim \mathbb{P}_{y}}\left[\left\|{D}_{S A}(\mathbf{\tilde{x}})-{D}_{S A}(\mathbf{y})\right\|_2^2\right].
\end{equation}
where, $\mathbf{y}$ is ground truth DEM and $\mathbf{\tilde{x}}$ is bicubic interpolated coarse SRTM DEM as mentioned in \S 3. The pixel loss $(\mathscr{L}_{P})$ and SSIM loss $(\mathscr{L}_{SSIM})$ and adversarial loss $(\mathscr{L}_{ADV})$ described in \S 3.4 are defined as,
\begin{equation} \label{eq88}
\begin{aligned}
&\mathscr{L}_{P}=\mathbb{E}_{\mathbf{\tilde{x}} \sim \mathbb{P}_{\tilde{x}}, \mathbf{z} \sim \mathbb{P}_{Z}, \mathbf{y} \sim \mathbb{P}_{y}}\left[\left\|\mathbf{y}-{G}(\mathbf{\tilde{x}}, \mathbf{z}\odot A_s(\mathbf{\tilde{x}}))\right\|_2^2\right],\\
&\mathscr{L}_{SSIM} = \mathbb{E}_{\mathbf{\tilde{x}} \sim \mathbb{P}_{\tilde{x}}, \mathbf{z} \sim \mathbb{P}_{Z}, y \sim \mathbb{P}_{y}} -\log ({SSIM}({G}(\mathbf{\tilde{x}}, \mathbf{z}\odot A_s(\mathbf{\tilde{x}})), \mathbf{y})),\\
&\mathscr{L}_{ADV} = \mathbb{E}_{\mathbf{\tilde{x}} \sim \mathbb{P}_{\tilde{x}}, \mathbf{z} \sim \mathbb{P}_{Z}} -\log ({D}({G}(\mathbf{\tilde{x}}, \mathbf{z}\odot A_s(\mathbf{\tilde{x}})))).
\end{aligned}
\end{equation}
where, $A_s(\mathbf{\tilde{x}}) = PSA({D}_{SA}(\mathbf{\tilde{x}}))$ with PSA being polarized self-attention as discussed in \S 3.3.

\section{Proof of Theorem 1: Smoothness of Sinkhorn Loss}\label{appendA}
We will define some of the terminologies, which are necessary for this proof. For all the proofs, we assume, $\mathbf{x} = concat(\mathbf{\tilde{x}}, \mathbf{z}\odot A_s(\mathbf{\tilde{x}}))$. From equation 6 of the main paper, the entropic optimal transport \cite{505} can be defined as, 
\begin{equation} \label{sup:eq1}
\begin{aligned}
&\mathcal{W}_{C,\varepsilon}\left(\mu_{\theta},\nu\right) = 
\inf_{\pi \in \Pi\left(\mu_{\theta},\nu\right)} \int_{\mathcal{X} \times \mathcal{Y}} [C\left(G_{\theta}\left(\mathbf{x}\right),\mathbf{y}\right)] d\pi\left(G_{\theta}\left(\mathbf{x}\right),\mathbf{y}\right)+ \varepsilon I_{\pi}\left(G_{\theta}\left(\mathbf{x}\right),\mathbf{y}\right),\\
& \text{where } I_{\pi}\left(G_{\theta}\left(\mathbf{x}\right),\mathbf{y}\right)) = \int_{\mathcal{X} \times \mathcal{Y}}[\log\left(\frac{\pi\left(G_{\theta}\left(\mathbf{x}\right),\mathbf{y}\right)}{\mu_{\theta}\left(G_{\theta}\left(\mathbf{x}\right)\right)\nu\left(\mathbf{y}\right)}\right)]d\pi\left(G_{\theta}\left(\mathbf{x}\right),\mathbf{y}\right),\\
 & \text{s.t.} \int_{\mathcal{X}} \pi\left(G_{\theta}\left(\mathbf{x}\right),\mathbf{y}\right) dx = \nu\left(\mathbf{y}\right), \: \int_{\mathcal{Y}} \pi\left(G_{\theta}\left(\mathbf{x}\right),\mathbf{y}\right) dy = \mu_{\theta}\left(G_{\theta}\left(\mathbf{x}\right)\right) \: \& \: \pi\left(G_{\theta}\left(\mathbf{x}\right),\mathbf{y}\right) \geq 0.
\end{aligned}
\end{equation}  
The formulation in equation \ref{sup:eq1} corresponds to the primal problem of regularized OT and, this allows us to express the dual formulation of regularized OT as the maximization of an expectation problem, as shown in equation \ref{sup:eq2} \cite{505}.
\begin{equation}\label{sup:eq2}
\begin{aligned}
\mathcal{W}&_{C,\varepsilon}\left(\mu_{\theta},\nu\right) = 
\sup_{\phi, \psi \in \mathbf{\Phi}}  \int_{\mathcal{X}} \phi\left(G_{\theta}\left(\mathbf{x}\right)\right) \, d\mu_{\theta}\left(G_{\theta}\left(\mathbf{x}\right)\right) + \int_{\mathcal{Y}} \psi\left(\mathbf{y}\right) \, d\nu\left(\mathbf{y}\right) \\ 
&- \varepsilon \int_{\mathcal{X}\times \mathcal{Y}} e^{\left(\frac{\phi\left(G_{\theta}\left(\mathbf{x}\right)\right) +  \psi\left(\mathbf{y}\right) - C\left(G_{\theta}\left(\mathbf{x}\right),\mathbf{y}\right)}{\varepsilon}\right)} \, d\mu_{\theta}\left(G_{\theta}\left(\mathbf{x}\right)\right)d\nu\left(\mathbf{y}\right)  + \varepsilon\\
\end{aligned}
\end{equation}
where $\mathbf{\Phi} = \{\left(\phi,\psi\right) \in \mathcal{C}\left(\mathcal{X}\right) \times \mathcal{C}\left(\mathcal{Y}\right)\}$ is set of real valued continuous functions for domain $\mathcal{X}$ and $\mathcal{Y}$ and they are referred as dual potentials. Now, given optimal dual potentials $\phi^{*}\left(\cdot\right)$, and $\psi^{*}\left(\cdot\right)$, the optmal coupling $\pi^{*}\left(\cdot\right)$ as per \cite{505} can be defined as
\begin{equation}\label{sup:eq3}
    \pi^{*}\left(G_{\theta}\left(\mathbf{x}\right),\mathbf{y}\right) = \mu_{\theta}\left(G_{\theta}\left(\mathbf{x}\right)\right) \nu\left(\mathbf{y}\right) e^{\frac{\phi^{*}\left(G_{\theta}\left(\mathbf{x}\right)\right) + \psi^{*}\left(\mathbf{y}\right) - C\left(G_{\theta}\left(\mathbf{x}\right),\mathbf{y}\right)}{\varepsilon}}.
\end{equation}
 To prove \textbf{Theorem 1}, we need an important property regarding its Lipschitz continuity of the dual potentials, which is explained in the following \textbf{Lemma}.
\begin{lemma}
   \textit{If $C\left(\cdot\right)$ is $L_0$ Lipschitz, then the dual potentials are also $L_0$ Lipschitz.}
\end{lemma}
\begin{proof}
    Assuming $\mathbf{\hat{y}} = G_{\theta}\left(\mathbf{x}\right)$, then $C\left(\mathbf{\hat{y}},\mathbf{y}\right)$ is $L_0$-Lipschitz in $\mathbf{\hat{y}}$. As, the entropy $I_{\pi}\left(\cdot\right)$ is selected as Shannon entropy, according to \cite{33} using the softmin operator, the optimal potential $\phi^{*}\left(\cdot\right)$ satisfy the following equation
    \begin{equation}\label{sup:eq4}
        \phi^{*}\left(\mathbf{\hat{y}}\right) = -\varepsilon \ln{\left[\int_{\mathcal{Y}} \exp{\left(\frac{\psi^{*}\left(\mathbf{y}\right)-C\left(\mathbf{\hat{y}},\mathbf{y}\right)}{\varepsilon}\right)}d\mathbf{y}\right]}
    \end{equation}
    Now, to estimate the Lipschitz of $\phi^{*}$, we have to find the upper bound of $||\nabla_{\mathbf{\hat{y}}} \phi^{*}\left(\mathbf{\hat{y}}\right)||$. Hence, taking the gradient of equation \ref{sup:eq4} with respect to $\mathbf{\hat{y}}$, the upper-bound of its norm can be written as,
  
    \begin{align}\label{sup:eq5}
         ||\nabla_{\mathbf{\hat{y}}} \phi^{*}\left(\mathbf{\hat{y}}\right)|| = \frac{||\int_{\mathcal{Y}} \exp{\left(\frac{\psi^{*}\left(\mathbf{y}\right)-C\left(\mathbf{\hat{y}},\mathbf{y}\right)}{\varepsilon}\right)}\nabla_{\mathbf{\hat{y}}} C\left(\mathbf{\hat{y}},\mathbf{y}\right)dy||}{||\int_{\mathcal{Y}} \exp{\left(\frac{\psi^{*}\left(\mathbf{y}\right)-C\left(\mathbf{\hat{y}},\mathbf{y}\right)}{\varepsilon}\right)}dy||}
    \end{align}
       Now due to Lipschitz continuity of $C\left(\mathbf{\hat{y}},\mathbf{y}\right)$, we can say $\nabla_{\mathbf{\hat{y}}} ||C\left(\mathbf{\hat{y}},\mathbf{y}\right)||\leq L_0$. Hence, using Cauchy-Schwarz inequality we will get,
       \begin{align}\label{sup:eq6}
         ||\nabla_{\mathbf{\hat{y}}} \phi^{*}\left(\mathbf{\hat{y}}\right)|| \leq ||\nabla_{\mathbf{\hat{y}}} C\left(\mathbf{\hat{y}},\mathbf{y}\right)||
         \frac{||\int_{\mathcal{Y}} \exp{\left(\frac{\psi^{*}\left(\mathbf{y}\right)-C\left(\mathbf{\hat{y}},\mathbf{y}\right)}{\varepsilon}\right)dy||}}{||\int_{\mathcal{Y}} \exp{\left(\frac{\psi^{*}\left(\mathbf{y}\right)-C\left(\mathbf{\hat{y}},\mathbf{y}\right)}{\varepsilon}\right)}dy||} = L_0.
    \end{align}
This completes the proof of the lemma. An alternative proof is provided by \cite{511} in Proposition 4. Similarly, it can be proved for the other potential term.
\end{proof}
For any $\theta_1$, $\theta_2\in \Theta$ will result in different coupling solutions $\pi^{*}_{i}$, for $i=1,2$. Now, based on Danskins' theorem for optimal coupling $\pi^{*}\left(\theta\right)$, we can write
\begin{align}\label{sup:eq7}
    \nabla_{\theta} \mathcal{W}_{C,\varepsilon}\left(\mu_{\theta},\nu\right) = \mathbb{E}_{G_{\theta}\left(\mathbf{x}\right),\mathbf{y} \sim \pi^{*}\left(\theta\right)}\left[\nabla_\theta C\left(G_{\theta}\left(\mathbf{x}\right),\mathbf{y}\right)\right]
\end{align}
Therefore, for any $\theta_1$ and $\theta_2$, we can write,
\begin{equation}
\begin{aligned}\label{sup:eq77}
    &||\nabla_{\theta} \mathcal{W}_{C,\varepsilon}\left(\mu_{\theta_1},\nu\right) - \nabla_{\theta} \mathcal{W}_{C,\varepsilon}\left(\mu_{\theta_2},\nu\right)|| \leq \\
    & ||\mathbb{E}_{G_{\theta_1}\left(\mathbf{x}\right),\mathbf{y} \sim \pi^{*}_1}\left[\nabla_\theta C\left(G_{\theta_1}\left(\mathbf{x}\right),\mathbf{y}\right)\right] - \mathbb{E}_{G_{\theta_1}\left(\mathbf{x}\right),\mathbf{y} \sim \pi^{*}_2}\left[\nabla_\theta C\left(G_{\theta_1}\left(\mathbf{x}\right),\mathbf{y}\right)\right]|| \\
    &+ ||\mathbb{E}_{G_{\theta_1}\left(\mathbf{x}\right),\mathbf{y} \sim \pi^{*}_2}\left[\nabla_\theta C\left(G_{\theta_1}\left(\mathbf{x}\right),\mathbf{y}\right)\right] - \mathbb{E}_{G_{\theta_2}\left(\mathbf{x}\right),\mathbf{y} \sim \pi^{*}_2}\left[\nabla_\theta C\left(G_{\theta_2}\left(\mathbf{x}\right),\mathbf{y}\right)\right]||\\
    & \leq L_0L||\pi^{*}_1 - \pi^{*}_2|| + L_1L||\theta_1 - \theta_2||
\end{aligned}
\end{equation}
Now with respect to different $\theta_i$, for $i=1,2$ with different pair of dual potentials, the $||\pi^{*}_1 - \pi^{*}_2||$ can be written as below. For simplicity we denote $\mu_\theta \equiv \mu_{\theta}\left(G_{\theta}\left(\mathbf{x}\right)\right)$ and $\nu \equiv \nu\left(\mathbf{y}\right)$.
\begin{equation}\label{sup:eq8}
    \begin{aligned}
    ||\pi^{*}_1 - \pi^{*}_2|| &= ||\mu_{\theta_1}\nu \exp{\left(\frac{\phi^{*}\left(G_{\theta_1}\left(\mathbf{x}\right)\right) + \psi^{*}\left(\mathbf{y}\right) - C\left(G_{\theta_1}\left(\mathbf{x}\right),\mathbf{y}\right)}{\varepsilon}\right)} \\
    &- \mu_{\theta_2}\nu \exp{\left(\frac{\phi^{*}\left(G_{\theta_2}\left(\mathbf{x}\right)\right) + \psi^{*}\left(\mathbf{y}\right) - C\left(G_{\theta_2}\left(\mathbf{x}\right),\mathbf{y}\right)}{\varepsilon}\right)}||  \\
  &\leq ||\nu \exp{\left(\frac{\phi^{*}\left(G_{\theta_1}\left(\mathbf{x}\right)\right) + \psi^{*}\left(\mathbf{y}\right) - C\left(G_{\theta_1}\left(\mathbf{x}\right),\mathbf{y}\right)}{\varepsilon}\right)}\left(\mu_{\theta_1}-\mu_{\theta_2}\right)|| \\
  & + ||\mu_{\theta_2}\nu \left[\exp{\left(\frac{\phi^{*}\left(G_{\theta_1}\left(\mathbf{x}\right)\right) + \psi^{*}\left(\mathbf{y}\right) - C\left(G_{\theta_1}\left(\mathbf{x}\right),\mathbf{y}\right)}{\varepsilon}\right)}\right.\\
  &-\left.\exp{\left(\frac{\phi^{*}\left(G_{\theta_2}\left(\mathbf{x}\right)\right) + \psi^{*}\left(\mathbf{y}\right) - C\left(G_{\theta_2}\left(\mathbf{x}\right),\mathbf{y}\right)}{\varepsilon}\right)} \right]|| 
\end{aligned}
\end{equation}
From \cite{34}, we know, as the dual potentials are $L_0$-Lipschitz, $\forall G_{\theta}\left(\mathbf{x}\right)\in \mathcal{X}$, we can write, $\phi^{*}\left(G_{\theta}\left(\mathbf{x}\right)\right)\leq L_0|G_{\theta}\left(\mathbf{x}\right)|$. And from property of c-transform, for $\forall \mathbf{y}\in \mathcal{Y}$ we can also write $\psi^{*}\left(\mathbf{y}\right)\leq \max_{G_{\theta}\left(\mathbf{x}\right)} \phi^{*}\left(G_{\theta}\left(\mathbf{x}\right)\right) - C\left(G_{\theta}\left(\mathbf{x}\right),\mathbf{y}\right)$. We assume $\mathcal{X}$ to be a bounded set in our case, hence, denoting $|\mathcal{X}|$ as the diameter of the space, at optimality, we can get that $\forall G_{\theta}\left(\mathbf{x}\right)\in \mathcal{X}, \: \mathbf{y}\in \mathcal{Y}$
\begin{equation}\label{sup:eq10}
\begin{aligned}
    &\Rightarrow \phi^{*}\left(G_{\theta}\left(\mathbf{x}\right)\right) + \psi^{*}\left(\mathbf{y}\right) \leq 2L_0|\mathcal{X}| + ||C||_{\infty}\\
    &\Rightarrow \exp{\left(\frac{\phi^{*}\left(G_{\theta}\left(\mathbf{x}\right)\right) + \psi^{*}\left(\mathbf{y}\right) - C\left(G_{\theta}\left(\mathbf{x}\right),\mathbf{y}\right)}{\varepsilon}\right))} \leq \exp{\left(2\frac{L_0|\mathcal{X}| + ||C||_{\infty}}{\varepsilon}\right)}
\end{aligned}
\end{equation}
 Hence, the exponential terms in equation \ref{sup:eq8} are bounded, and we can assume it has a finite Lipschitz constant $L_{exp}$. Taking $\kappa = 2\left(L_0|\mathcal{X}| + ||C||_{\infty}\right)$, and using Cauchy-Schwarz, we can rewrite equation \ref{sup:eq8} as,
 \begin{equation}\label{sup:eq11}
\begin{aligned}
    ||\pi^{*}_1 - \pi^{*}_2|| &\leq \exp{\left(\frac{\kappa}{\varepsilon}\right)}||\nu||.||\mu_{\theta_1}-\mu_{\theta_2}|| \\
&+L_{exp}||\mu_{\theta_2}||.||\nu||.||\frac{\left(\phi^{*}\left(G_{\theta_1}\left(\mathbf{x}\right)\right)- \phi^{*}\left(G_{\theta_2}\left(\mathbf{x}\right)\right)\right)- \left(C\left(G_{\theta_1}\left(\mathbf{x}\right),\mathbf{y}\right)-C\left(G_{\theta_2}\left(\mathbf{x}\right),\mathbf{y}\right)\right)}{\varepsilon}||\\
&\leq \exp{\left(\frac{\kappa}{\varepsilon}\right)}||\nu||.||\mu_{\theta_1}-\mu_{\theta_2}|| + 2\frac{L_{exp}L_0L}{\varepsilon}||\mu_{\theta_2}||.||\nu||.||\theta_1-\theta_2|| 
\end{aligned}
\end{equation}
Now, as the input space $\mathcal{X}$ and output space $\mathcal{Y}$ are bounded, the corresponding measures $\mu_\theta$ and $\nu$ will also be bounded. We assume, $||\mu_\theta|| \leq \lambda_1$ and $||\nu||\leq \lambda_2$. If we apply equation \ref{sup:eq10} in equation \ref{sup:eq3}, to get the upper bound of the coupling function, we will get $||\pi^{*}_1 - \pi^{*}_2||\leq \exp{\left(\frac{\kappa}{\varepsilon}\right)}||\nu||.||\mu_{\theta_1}-\mu_{\theta_2}||$ which is less than the bound in equation \ref{sup:eq11}. Then, we can find some constant upper bound of $||\pi^{*}_1 - \pi^{*}_2||$, using the assumed bounds of measures and can write $||\pi^{*}_1 - \pi^{*}_2||\leq \exp{\left(\frac{\kappa}{\varepsilon}\right)}||\nu||.||\mu_{\theta_1}-\mu_{\theta_2}||\leq K$, such that, 
\begin{equation*}
    K \leq  \exp{\left(\frac{\kappa}{\varepsilon}\right)}||\nu||.||\mu_{\theta_1}-\mu_{\theta_2}|| + 2\frac{L_{exp}L_0L}{\varepsilon}||\mu_{\theta_2}||.||\nu||.||\theta_1-\theta_2||
\end{equation*}
Then using the marginal condition as shown in in equation \ref{sup:eq1}, we can write equation \ref{sup:eq11} as,
\begin{equation}\label{sup:eq12}
\begin{aligned}
    K & \leq \lambda_1\exp{\left(\frac{\kappa}{\varepsilon}\right)}||\int_{\mathcal{X}} \pi^{*}_1 d\mathbf{x} - \int_{\mathcal{X}} \pi^{*}_2 d\mathbf{x}|| + 2\lambda_1 \lambda_2 \frac{L_{exp}L_0L}{\varepsilon}||\theta_1-\theta_2||\\
    & \leq \lambda_1\exp{\left(\frac{\kappa}{\varepsilon}\right)} \int_{\mathcal{X}} ||\pi^{*}_1 - \pi^{*}_2||.|d\mathbf{x}| + 2\lambda_1 \lambda_2 \frac{L_{exp}L_0L}{\varepsilon}||\theta_1-\theta_2||\\
    & \leq \lambda_1\exp{\left(\frac{\kappa}{\varepsilon}\right)}K \int_{\mathcal{X}} |d\mathbf{x}| +2\lambda_1 \lambda_2 \frac{L_{exp}L_0L}{\varepsilon}||\theta_1-\theta_2|| 
\end{aligned}
\end{equation}
The input set is a compact set such that $\mathcal{X} \subset \mathbb{R}^d$. So, assuming $m$ and $M$ to be the minimum and maximum value in set $\mathcal{X}$ and considering the whole situation in discrete space, equation \ref{sup:eq12}, can be rewritten as,
\begin{equation}\label{sup:eq13}
\begin{aligned}
    K &\leq \lambda_1\exp{\left(\frac{\kappa}{\varepsilon}\right)}K \sum_{\mathbf{x}\in \mathcal{X}} |\mathbf{x}| + 2\lambda_1 \lambda_2 \frac{L_{exp}L_0L}{\varepsilon}||\theta_1-\theta_2||\\
    & \leq \lambda_1\exp{\left(\frac{\kappa}{\varepsilon}\right)}Kd\max\left(||M||,|||m|\right) + 2\lambda_1 \lambda_2 \frac{L_{exp}L_0L}{\varepsilon}||\theta_1-\theta_2||, 
\end{aligned}
\end{equation}
Now, taking $B = d\max\left(||M||,|||m|\right)$, and doing necessary subtraction and division on both sides of equation \ref{sup:eq13}, it can be rewritten as
\begin{equation}\label{sup:eq14}
    \begin{aligned}
        K & \leq \frac{2\lambda_1 \lambda_2L_{exp}L_0L}{\varepsilon\left(1-\lambda_1 B\exp{\left(\frac{\kappa}{\varepsilon}\right)}\right)}||\theta_1-\theta_2||\\
    & \leq \frac{2\lambda_1 \lambda_2L_{exp}L_0L}{\varepsilon\left(1+\lambda_1 B\exp{\left(\frac{\kappa}{\varepsilon}\right)}\right)}||\theta_1-\theta_2||
    \end{aligned}
\end{equation}
Equation \ref{sup:eq14}, satisfies because $\frac{\kappa}{\varepsilon}\geq0$. As, $||\pi^{*}_1 - \pi^{*}_2||\leq K$, from equation \ref{sup:eq14}, it can be written as 
\begin{equation}\label{sup:eq15}
    ||\pi^{*}_1 - \pi^{*}_2|| \leq \frac{2\lambda_1 \lambda_2L_{exp}L_0L}{\varepsilon\left(1+\lambda_1 B\exp{\left(\frac{\kappa}{\varepsilon}\right)}\right)}||\theta_1-\theta_2||
\end{equation}
Substituting equation \ref{sup:eq15} in equation \ref{sup:eq77}, we will get,
\begin{equation}\label{sup:eq16}
\begin{aligned}
    ||\nabla_{\theta} \mathcal{W}_{C,\varepsilon}\left(\mu_{\theta_1},\nu\right) - \nabla_{\theta} \mathcal{W}_{C,\varepsilon}\left(\mu_{\theta_2},\nu\right)|| &\leq L_0L||\pi^{*}_1 - \pi^{*}_2|| + L_1L||\theta_1 - \theta_2||\\
    & \leq \left(L_1L + \frac{2\lambda_1 \lambda_2L_{exp} L_0^2 L^2}{\varepsilon\left(1+\lambda_1 B\exp{\left(\frac{\kappa}{\varepsilon}\right)}\right)}\right)||\theta_1-\theta_2|| 
\end{aligned}
\end{equation}
So, the EOT problem defined in equation \ref{sup:eq1} has $\hat{\Gamma}_{\varepsilon}$ smoothness in $\theta$ with $\hat{\Gamma}_{\varepsilon} = L_1L + \frac{2\lambda_1 \lambda_2L_{exp} L_0^2 L^2}{\varepsilon\left(1+\lambda_1 B\exp{\left(\frac{\kappa}{\varepsilon}\right)}\right)}$. From this, we can derive the smoothness of Sinkhorn loss defined in equation 3 of main paper. Note that only the first two terms in this equation are $\theta$ dependent. Therefore, they only contribute to the gradient approximation and both of them will satisfy the same smoothness condition as defined in equation \ref{sup:eq16}. So, if Sinkhorn loss has smoothness $\Gamma_{\varepsilon}$, it will satisfy,
$\Gamma_{\varepsilon} = \frac{3}{2}\hat{\Gamma}_{\varepsilon}$. In general, we can define the smoothness of Sinkhorn loss with $\left(\theta_1, \theta_2\right) \in \Theta$ as,
\begin{equation}
\begin{aligned}
    ||\nabla_{\theta}S_{C,\varepsilon}\left(\mu_{\theta_1}, \nu\right) - \nabla_{\theta}S_{C,\varepsilon}\left(\mu_{\theta_2}, \nu\right)||\leq \mathcal{O} \left(L_1L + \frac{2L_0^2 L^2}{\varepsilon\left(1+B\exp{\left(\frac{\kappa}{\varepsilon}\right)}\right)}\right)||\theta_1-\theta_2||
\end{aligned}
\end{equation}
This completes the statement of \textbf{Theorem 1}

\section{Proof of proposition 1: Upper-bound of expected gradient in SiRAN set-up}\label{theorem2}
This proof is inspired by \cite{501}. Assuming $\Gamma = \mathcal{O}\left(L_1 + \frac{2L_0^2}{\varepsilon\left(1+B\exp{\left(\frac{\kappa}{\varepsilon}\right)}\right)}\right)$ be the smoothness in $p$ for Sinkhorn loss $S_{C,\varepsilon}\left(\mu_{\theta}\left(p\right), \nu\left(\mathbf{y}\right)\right)$, where $p=G_{\theta}\left(\mathbf{x}\right)$. For simplicity, we use a common set for inputs and outputs as $\mathcal{P}$. Hence, to approximate the gradient of Sinkhorn loss, using Jensen's inequality, we can write,
\begin{equation}
\begin{aligned}\label{sup:eq18}
    ||\nabla_{\theta} \mathbb{E}_{\left(\mathbf{x},\mathbf{y}\right)\sim \mathcal{P}}[S_{C,\varepsilon}\left(\mu_{\theta}\left(G_{\theta}\left(\mathbf{x}\right)\right), \nu\left(\mathbf{y}\right)\right)]|| &\leq \mathbb{E}_{\left(\mathbf{x},\mathbf{y}\right)\sim \mathcal{P}} \left[||\nabla_{\theta} S_{C,\varepsilon}\left(\mu_{\theta}\left(G_{\theta}\left(\mathbf{x}\right)\right), \nu\left(\mathbf{y}\right)\right)||\right]\\
     & \leq \mathbb{E}_{\left(\mathbf{x},\mathbf{y}\right)\sim \mathcal{P}} \left[ \underbrace{||\nabla_p S_{C,\varepsilon}\left(\mu_{\theta}\left(p\right), \nu\left(\mathbf{y}\right)\right)|| . ||\nabla_\theta G_{\theta}\left(\mathbf{x}\right)||}_\textrm{Cauchy-Schwarz inequality}\right]\\
     & \leq L\mathbb{E}_{\left(\mathbf{x},\mathbf{y}\right)\sim \mathcal{P}} \left[||\nabla_p S_{C,\varepsilon}\left(\mu_{\theta}\left(p\right), \nu\left(\mathbf{y}\right)\right)||\right] 
\end{aligned}
\end{equation}
Say, for optimized parameter $\theta^{*}$, $t= G_{\theta^{*}}(\mathbf{x})$. Since, $||\theta - \theta^{*}||$, we can write using the smoothness of sinkhorn loss and Lipschitz of model parameters,
\begin{equation}
\begin{aligned}\label{sup:eq19}
    &||\nabla_p S_{C,\varepsilon}\left(\mu_{\theta}\left(p\right), \nu\left(\mathbf{y}\right)\right)|| - ||\nabla_t S_{C,\varepsilon}\left(\mu_{\theta^{*}}\left(t\right), \nu\left(\mathbf{y}\right)\right)|| \\
    &\leq ||\nabla_p S_{C,\varepsilon}\left(\mu_{\theta}\left(p\right), \nu\left(\mathbf{y}\right)\right) - \nabla_t S_{C,\varepsilon}\left(\mu_{\theta^{*}}\left(t\right), \nu\left(\mathbf{y}\right)\right)||\\
    & \leq \Gamma ||p-t|| = \Gamma ||G(\theta)(\mathbf{x}) - G(\theta^{*}(\mathbf{x}))||\\
    & \leq \Gamma L ||\theta - \theta^{*}|| \leq \Gamma L \epsilon 
\end{aligned}
\end{equation}
At optimal condition, $||\nabla_t S_{C,\varepsilon}\left(\mu_{\theta^{*}}\left(t\right), \nu\left(\mathbf{y}\right)\right)|| = 0$ as the distributions of $y$ and $t = G_{\theta^*}(\mathbf{x})$ are aligned for optimal $\theta^*$. So, by substituting equation \ref{sup:eq19} in equation \ref{sup:eq18}, we will get
\begin{equation}
\begin{aligned}\label{sup:eq20}
    ||\nabla_{\theta} \mathbb{E}_{\left(\mathbf{x},\mathbf{y}\right)\sim \mathcal{P}}\left[S_{C,\varepsilon}\left(\mu_{\theta}\left(G_{\theta}\left(\mathbf{x}\right)\right), \nu\left(\mathbf{y}\right)\right)\right]||\leq L^2\Gamma\epsilon 
\end{aligned}
\end{equation}
From Lemma 1 of \cite{501}, we get,
\begin{equation}
\begin{aligned}\label{sup:eq21}
    ||\nabla_{\theta} \mathbb{E}_{\left(\mathbf{x},\mathbf{y}\right)\sim \mathcal{P}} \left[l(G_{\theta}(\mathbf{x}),\mathbf{y})\right]||\leq L^2\beta\epsilon
\end{aligned}    
\end{equation}
Similarly, from Lemma 2 of \cite{501}, we get
\begin{equation}
\begin{aligned}\label{sup:eq22}
    ||-\nabla_{\theta} \mathbb{E}_{\left(\mathbf{x},\mathbf{y}\right)\sim \mathcal{P}} \left[g(\psi;G_{\theta}(\mathbf{x}))\right]||\leq L\delta 
\end{aligned}    
\end{equation}
Here, $\psi$ is parameters of discriminator $D$. So using equations \ref{sup:eq20}, \ref{sup:eq21}, and \ref{sup:eq22}, for the combination of losses we will get,
\begin{equation}
    \begin{aligned}
    & ||\nabla_{\theta} \mathbb{E}_{\left(\mathbf{x},\mathbf{y}\right)\sim \mathcal{P}} [l(G_{\theta}(\mathbf{x}),\mathbf{y}) + S_{C,\varepsilon}\left(\mu_{\theta}\left(G_{\theta}\left(\mathbf{x}\right)\right), \nu\left(\mathbf{y}\right)\right)
    -g(\psi;G_{\theta}(\mathbf{x}))]|| \\&\leq ||\nabla_{\theta} \mathbb{E}_{\left(\mathbf{x},\mathbf{y}\right)\sim \mathcal{P}} \left[l(G_{\theta}(\mathbf{x}),\mathbf{y})\right]|| + ||\nabla_{\theta} \mathbb{E}_{\left(\mathbf{x},\mathbf{y}\right)\sim \mathcal{P}}\left[S_{C,\varepsilon}\left(\mu_{\theta}\left(G_{\theta}\left(\mathbf{x}\right)\right), \nu\left(\mathbf{y}\right)\right)\right]|| \\
    &+  ||-\nabla_{\theta} \mathbb{E}_{\left(\mathbf{x},\mathbf{y}\right)\sim \mathcal{P}} \left[g(\psi;G_{\theta}(\mathbf{x}))\right]||\\
    &\leq L^2\beta\epsilon + L^2\Gamma\epsilon + L\delta = L^2\epsilon(\beta + \Gamma) +  L\delta
\end{aligned}
\end{equation}
This completes the proof.

\section{Proof of Proposition2: Iteration complexity of SiRAN}\label{theorem3}
This proof also follows the steps of Theorem 3 from \cite{501}. In the sinkhorn regularized adversarial setup, the parameters $\theta$ are updated using fixed step gradient descent. They iterate as,
\begin{equation}
 \begin{aligned}\label{sup:eq23}
     \theta_{t+1} = \theta_t - h_t \nabla (l(\theta_t)+S_{C,\varepsilon}\left(\mu_{\theta_t}\left(G_{\theta_t}\left(\mathbf{x}\right)\right), \nu\left(\mathbf{y}\right)\right) - g(\psi;G_{\theta_t}(\mathbf{x}))).
 \end{aligned}   
\end{equation}
 For simplicity, we denote $S_{C,\varepsilon}\left(\mu_{\theta_t}\left(G_{\theta_t}\left(\mathbf{x}\right)\right), \nu\left(\mathbf{y}\right)\right) \equiv S_{C,\varepsilon}(\mu_{\theta_t}, \nu)$. Using Taylor's expansion,
\begin{equation}
\begin{aligned}\label{sup:eq24}
    l(\theta_{t+1}) = l(\theta_t) + \nabla l(\theta_t)(\theta_{t+1} - \theta_t) 
    + \frac{1}{2}(\theta_{t+1} - \theta_t)^T \nabla^2 l(\theta_t)(\theta_{t+1} - \theta_t)
\end{aligned}    
\end{equation}
Now, substituting $\theta_{t+1} - \theta_t$ from equation \ref{sup:eq23}, and using triangle inequality and Cauchy-Schwarz inequality, equation \ref{sup:eq24} can be rewritten as,
\begin{equation}
\begin{aligned}\label{sup:eq25}
     l(\theta_{t+1}) \leq &
    l(\theta_t)\!- h_t||\nabla l(\theta_t)||^2\! - h_t||\nabla l(\theta_t)||.||\nabla S_{C,\varepsilon}(\mu_{\theta_t}, \nu)|| - h_t||\nabla l(\theta_t)||.||g(\psi;G_{\theta_t}(\mathbf{x}))|| \\
    &+ h_t^2||\nabla(l(\theta_t) + S_{C,\varepsilon}(\mu_{\theta_t}, \nu) - g(\psi;G_{\theta_t}(\mathbf{x})))||^2 \frac{||\nabla^2 l(\theta_t)||}{2}.
    \end{aligned}    
\end{equation}
Taking into account the assumptions in \textbf{Propositions 2} and utilizing Minkowski's inequality, equation \ref{sup:eq25} can be rewritten as,
\begin{equation}
\begin{aligned}\label{sup:eq26}
    l(\theta_{t+1}) \leq &
    l(\theta_t) - h_t||\nabla l(\theta_t)||^2 - h_t||\nabla l(\theta_t)||\eta 
    -  h_t||\nabla l(\theta_t)||\zeta \\ 
    & + h_t^2 (||\nabla(l(\theta_t)||^2 + ||S_{C,\varepsilon}(\mu_{\theta_t}, \nu)||^2 
     + ||g(\psi;G_{\theta_t}(\mathbf{x})))||^2) \frac{\beta_1}{2}.
\end{aligned}    
\end{equation}
Using $h_t = \frac{1}{\beta_1}$, from equation \ref{sup:eq26}, we can write,
\begin{equation}
\begin{aligned}
    l(\theta_{t+1}) &\leq  l(\theta_t) - \frac{h_t||\nabla l(\theta_t)||^2}{2} - h_t||\nabla l(\theta_t)||\eta -  h_t||\nabla l(\theta_t)||\zeta \\
    &+ \frac{h_t||S_{C,\varepsilon}(\mu_{\theta_t}, \nu)||^2}{2} + \frac{h_t||g(\psi;G_{\theta_t}(\mathbf{x})))||^2}{2}\\
    & \leq l(\theta_t) - \frac{h_t \epsilon_1^2}{2} - h_t\epsilon_1 \eta -h_t\epsilon_1\zeta
    + \frac{h_tL^4\Gamma^2\epsilon^2}{2}+\frac{h_tL^2\delta^2}{2}. 
\end{aligned}    
\end{equation}
Assuming T iterations to reach this $\epsilon_1$-stationary point, then for $t\leq T$, doing telescopic sum over $t$,
\begin{equation}
\begin{aligned}
    \sum_{t=0}^{T-1} l(\theta_{t+1}) -&  l(\theta_{t})  \leq \frac{-T(\epsilon_1^2+2\epsilon_1(\zeta+\eta) - L^2(\delta^2+L^2\Gamma^2\epsilon^2))}{2\beta_1}\\
    &\Rightarrow T\leq \frac{2(l(\theta_0) - l^*)\beta_1}{(\epsilon_1^2+2\epsilon_1(\zeta+\eta) - L^2(\delta^2+L^2\Gamma^2\epsilon^2))}
\end{aligned}    
\end{equation}
Therefore, using the iteration complexity definition of \cite{501}, we obtain,
\begin{equation}
\begin{aligned}
    \sup_{\theta_0\in \{\mathbb{R}^{h\times d_x},\mathbb{R}^{d_\mathbf{y}\times h}\},l\in \mathscr{L}} \mathcal{T}_{\epsilon_1}(A_h[l,\theta_0],l) 
    = \mathcal{O}\left(\frac{(l(\theta_0) - l^*)\beta_1}{\epsilon_1^2+2\epsilon_1(\zeta+\eta) - L^2(\delta^2+L^2\Gamma_{\varepsilon}^2\epsilon^2)}\right).
\end{aligned}
\end{equation}
This completes the proof of \textbf{Proposition 2}.
\subsection{Proof of Corollary 1}\label{appendC.1}
Using the similar arguments of \textbf{Proposition 2}, and taking first-order Taylor's approximation, we get
\begin{equation}
\begin{aligned}
    l(\theta_{t+1}) &= l(\theta_{t}) - h_t||\nabla l(\theta_t)||^2 - h_t||\nabla l(\theta_t)||.||\nabla S_{C,\varepsilon}(\mu_{\theta_t}, \nu)|| - h_t||\nabla l(\theta_t)||.||g(\psi;G_{\theta_t}(\mathbf{x}))||\\
    & \leq l(\theta_t) - h_t \epsilon_1^2 - h_t\epsilon_1 \eta -h_t\epsilon_1\zeta
\end{aligned}    
\end{equation}
Taking telescopic sum over $t$ for $t\leq T$, we get 
\begin{equation}
     \sum_{t=0}^{T-1} l(\theta_{t+1}) - l(\theta_{t}) \leq -Th_t(\epsilon_1^2+\epsilon_1(\zeta+\eta))
\end{equation}
So, using the definition of iteration complexity, we get,
\begin{equation}
\begin{aligned}
    \sup_{\theta_0\in \{\mathbb{R}^{h\times d_x},\mathbb{R}^{d_\mathbf{y}\times h}\},l\in \mathscr{L}} \mathcal{T}_{\epsilon_1}(A_h[l,\theta_0],l) 
    = \mathcal{O}\left(\frac{l(\theta_0) - l^*}{\epsilon_1^2+\epsilon_1(\zeta+\eta)}\right)
\end{aligned}
\end{equation}
This completes the proof.

\section{Empirical results to prove Proposition 1 and Proposition 2}\label{final}
We perform experiments to answer the proposed claims. There are two main aspects we want to investigate, firstly, how the choice of $\varepsilon$ affects the overall training of the model, and secondly, how it performs compared to other state-of-the-art learning methods like WGAN, WGAN+GP, and DCGAN. In both these cases,  we analyze the claims of mitigating vanishing gradients in the near-optimal region and fast convergence rate.
\subsection{Experiment set-up}
In this setting, we are performing a denoising operation on the MNIST dataset. For this 60000 samples of size $28\times 28$ are used during training, while 10000 are used for testing. The convergence criterion is set to be the mean square error of 0.04
or a maximum of 500 epochs. During training, we randomly add Gaussian noise to the training samples to perform the denoising task. The generator is designed as a simple autoencoder structure with an encoder and decoder each having 2 convolutional layers.
In practice, we notice that a discriminator with shallow layers is usually sufficient to offer a higher convergence rate. Therefore, we choose, a three-layer fully connected network with 1024 and 256 hidden neurons. All the layers are followed by ReLu activation except the output layer. For optimization, ADAM is utilized with a learning rate of 0.001 with a batch size of 64, and the discriminator is updated once for every single update of the generator.

\subsection{Result analysis}
Figure \ref{fig55}, shows how changing the value of $\varepsilon$ affects the overall iteration complexity. According to this figure, the instances $\varepsilon$ are very small and very large, and the learning behavior of the model becomes close to regular adversarial setup which ultimately results in more time requirement for convergence. This is because, as $\varepsilon \to 0$ and $\varepsilon \to \infty$, the smoothness of sinkhorn loss tends to become independent of $\varepsilon$ as depicted in \textbf{Theorem 1}, which makes the overall setup similar to the regular adversarial framework. This also affects the capability of mitigating the vanishing gradient problem as shown in Figure \ref{fig66} and \ref{fig77}. The gradients are approximated using spectral norm and they are moving averaged for better visualization. From Figure \ref{fig66}, in the case of the first layer, as $\varepsilon$ varies, the estimated gradients are similar near the optimal region. However, From Figure \ref{fig77}, we can see for the case of the hidden layer, gradient approximation is definitely affected by the choice of $\varepsilon$, and we can see as $\varepsilon \to 0$ and $\varepsilon \to \infty$, the gradients near-optimal region become smaller. However, using $\varepsilon=0.1$ tends to have higher gradients even if near the optimal region. Therefore, this model will have more capability of mitigating the vanishing gradient problem. Hence, we use this model to compare with other state-of-the-art learning methods.\par
We compare the rate of convergence and capability of handling the vanishing gradient of SIRAN with WGAN \cite{300}, WGAN+GP \cite{500}, and DCGAN. Figure \ref{fig555} clearly visualizes how our proposed framework has tighter iteration complexity than others, and reaches the convergence faster. This is consistent with the theoretical analysis presented in \textbf{Proposition 1}. Figure \ref{fig666} and \ref{fig777} also provides empirical evidence of the vanishing gradient issue presented in \textbf{Proposition 2}. Both for the first layer and hidden layer, as shown in Figure \ref{fig666} and \ref{fig777}, the approximated gradients are higher comparatively than others near the optimal region. This results in increasing the effectiveness of SIRAN in handling the issue of the vanishing gradient problem as discussed in above theorems.
\begin{figure}[!htb]
\centering
\minipage{0.32\textwidth}
  \includegraphics[width=\linewidth]{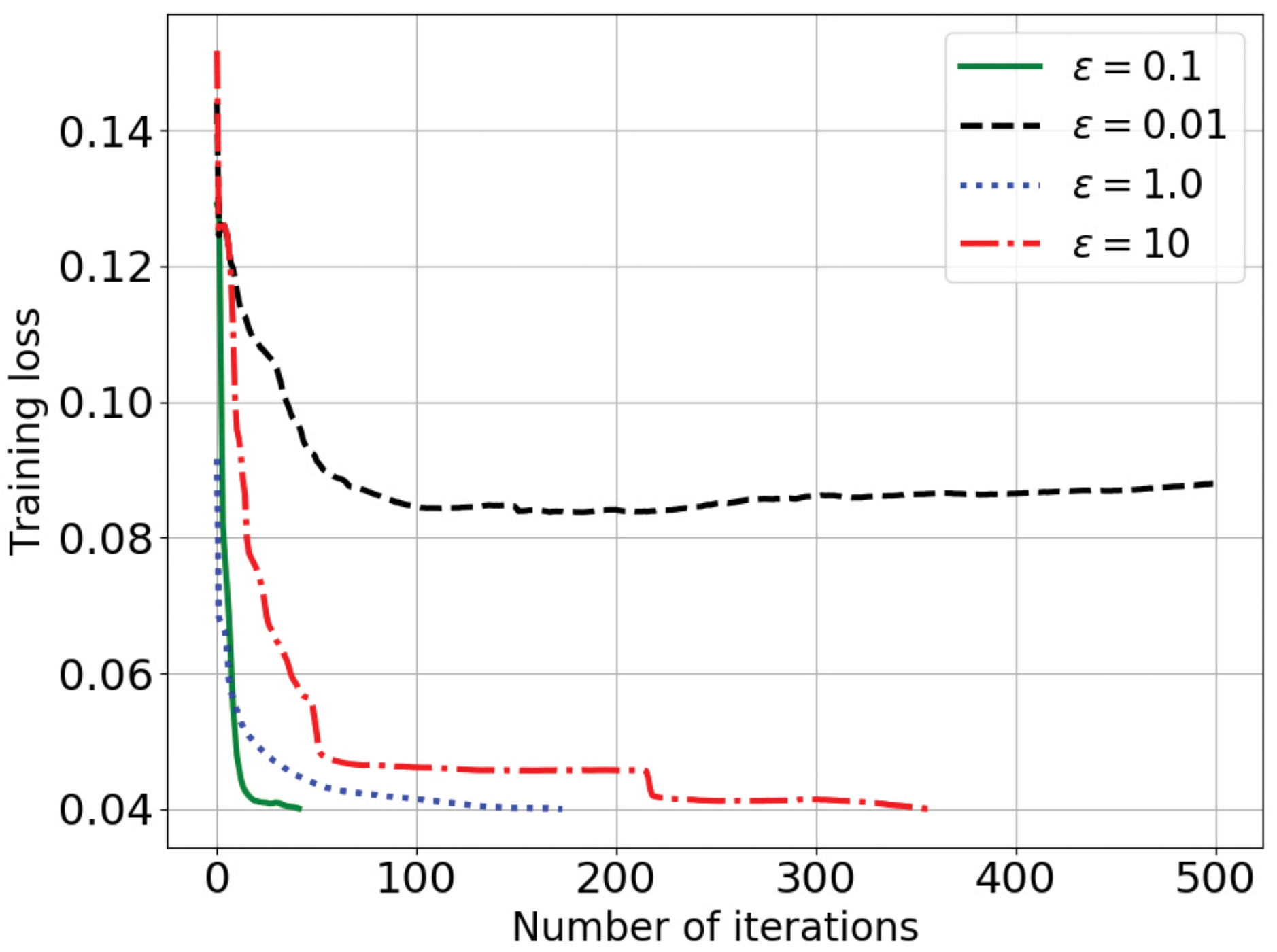}
  \caption{Training Loss for variation of $\varepsilon$}\label{fig55}
\endminipage \quad
\minipage{0.32\textwidth}
  \includegraphics[width=\linewidth]{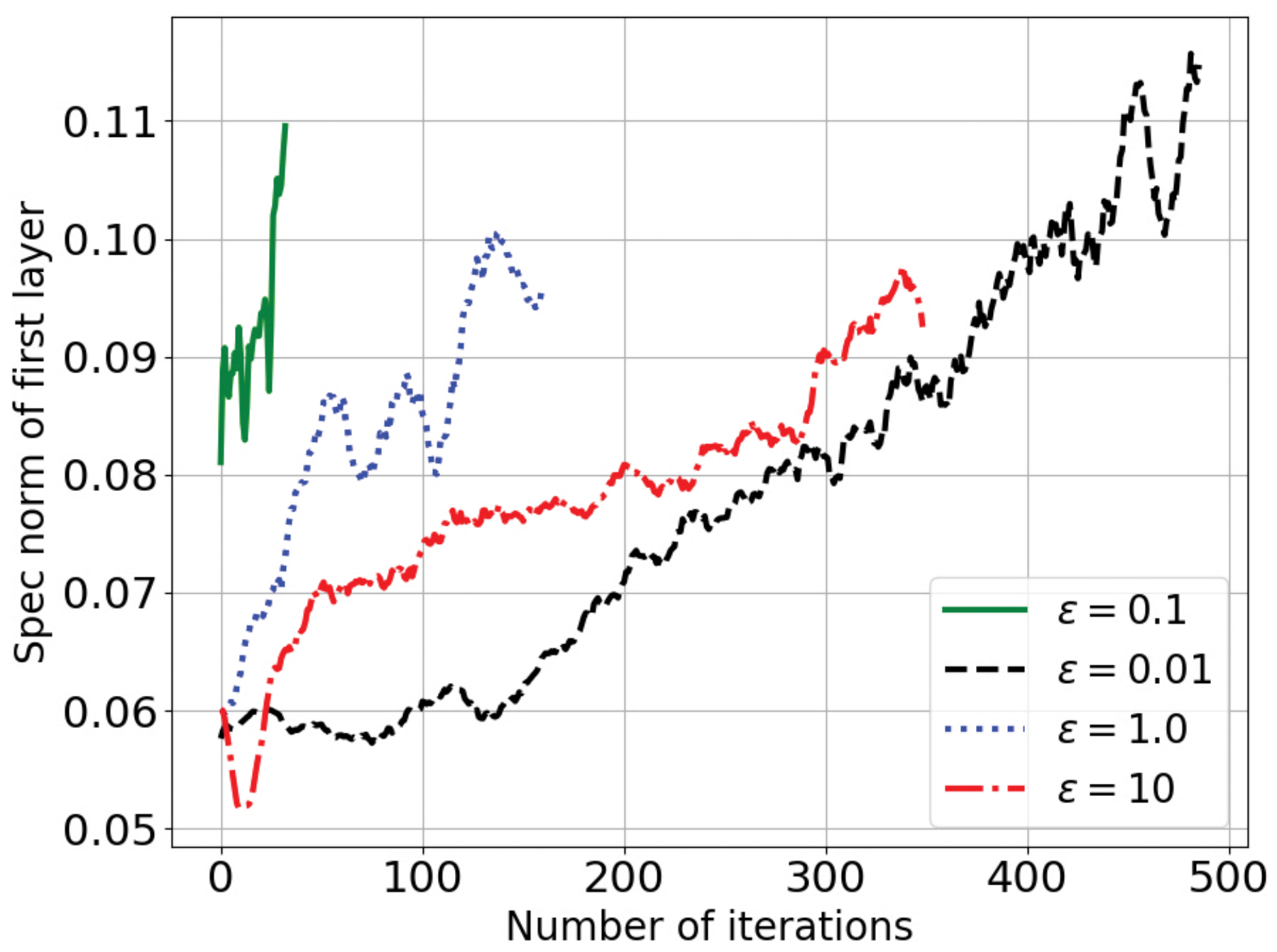}
  \caption{Approximated Spectral norm of gradients of first layer for different values of $\varepsilon$}\label{fig66}
\endminipage \quad
\minipage{0.3\textwidth}
  \includegraphics[width=\linewidth]{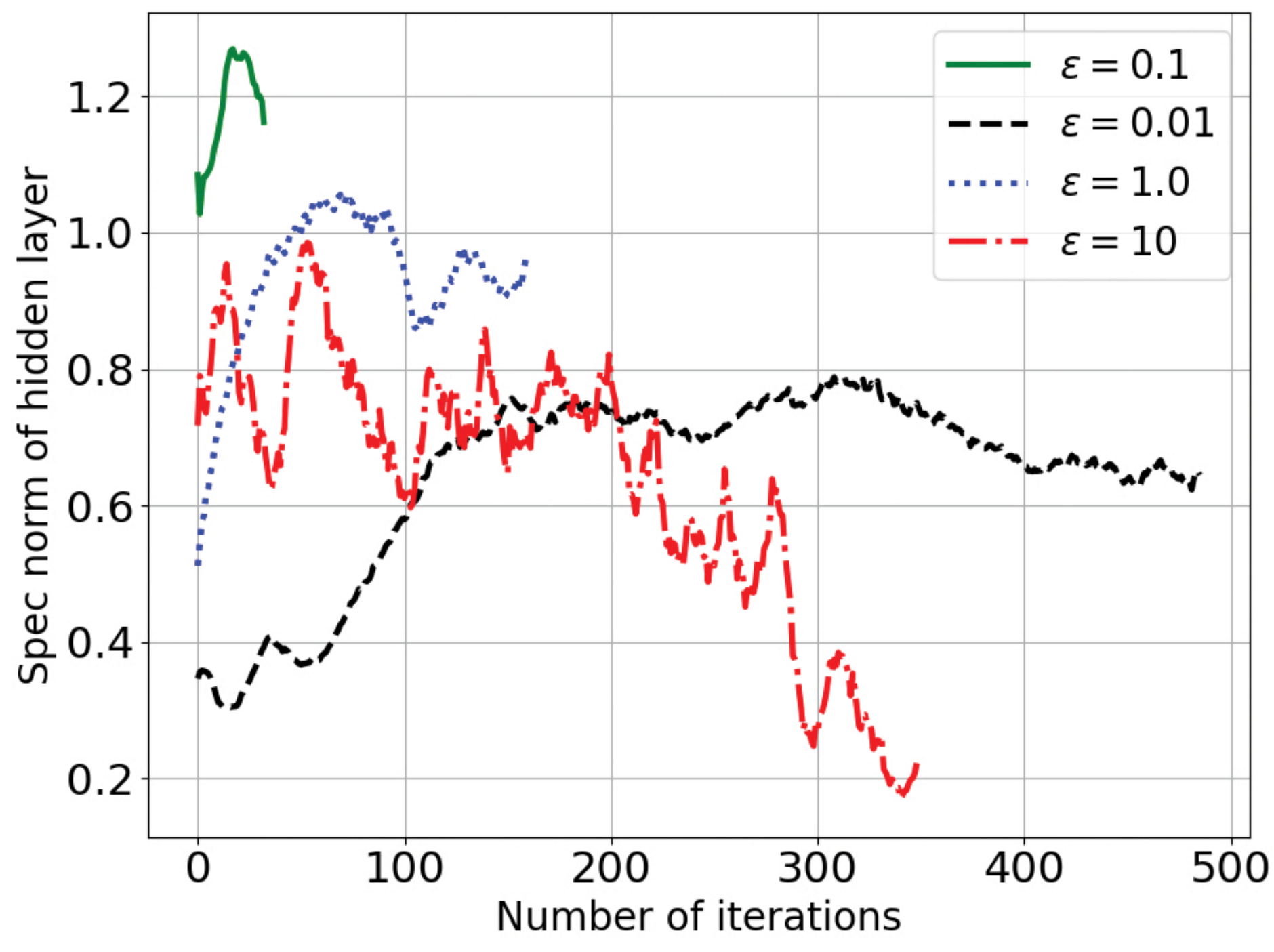}
  \caption{Approximated Spectral norm of gradients of hidden layer for different values of $\varepsilon$}\label{fig77}
  \endminipage
\end{figure}

\begin{figure}[!htb]
\centering
\minipage{0.32\textwidth}
  \includegraphics[width=\linewidth]{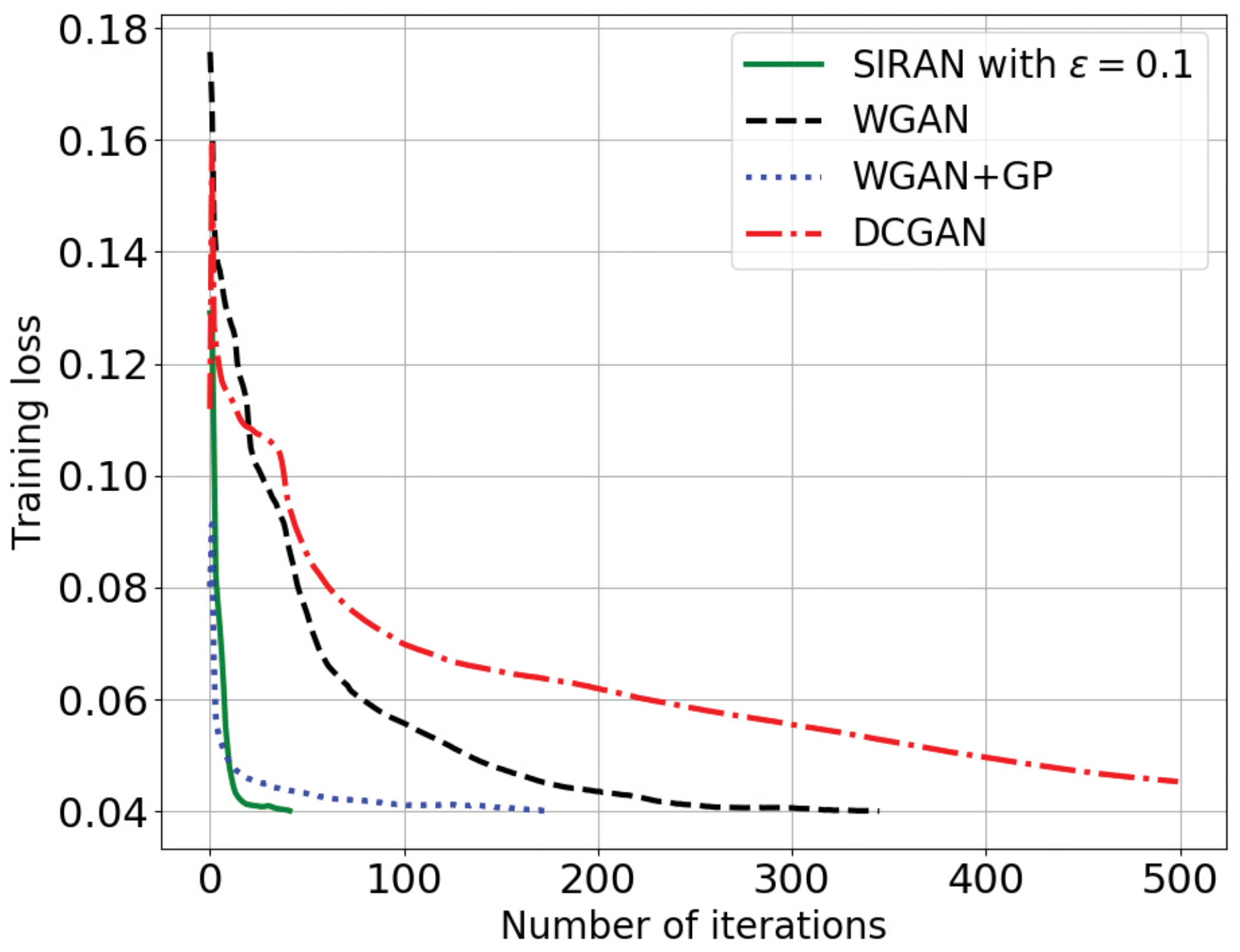}
  \caption{Training Loss for different learning methods}\label{fig555}
\endminipage \quad
\minipage{0.32\textwidth}
  \includegraphics[width=\linewidth]{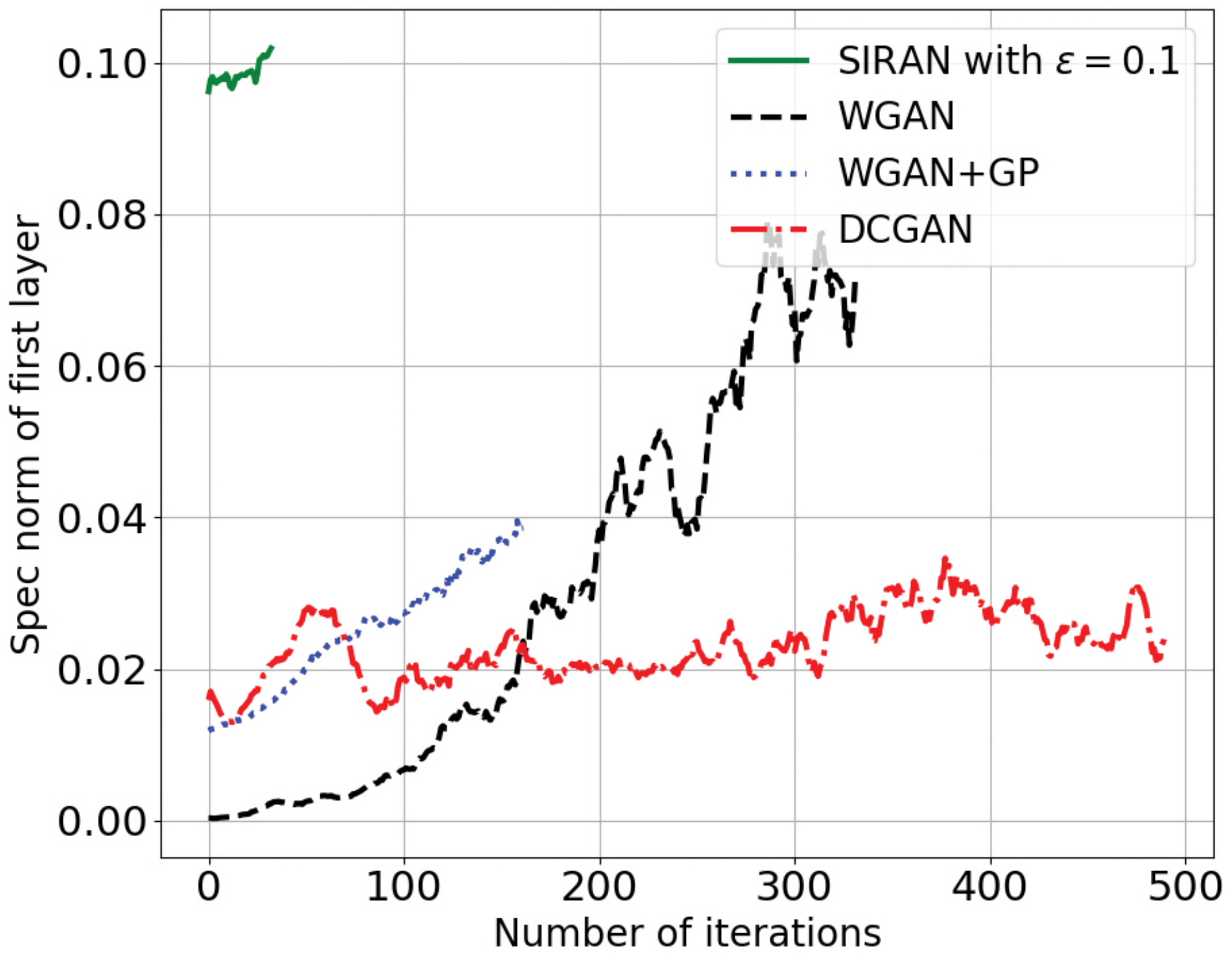}
  \caption{Approximated Spectral norm of gradients of first layer for different learning methods}\label{fig666}
\endminipage \quad
\minipage{0.3\textwidth}
  \includegraphics[width=\linewidth]{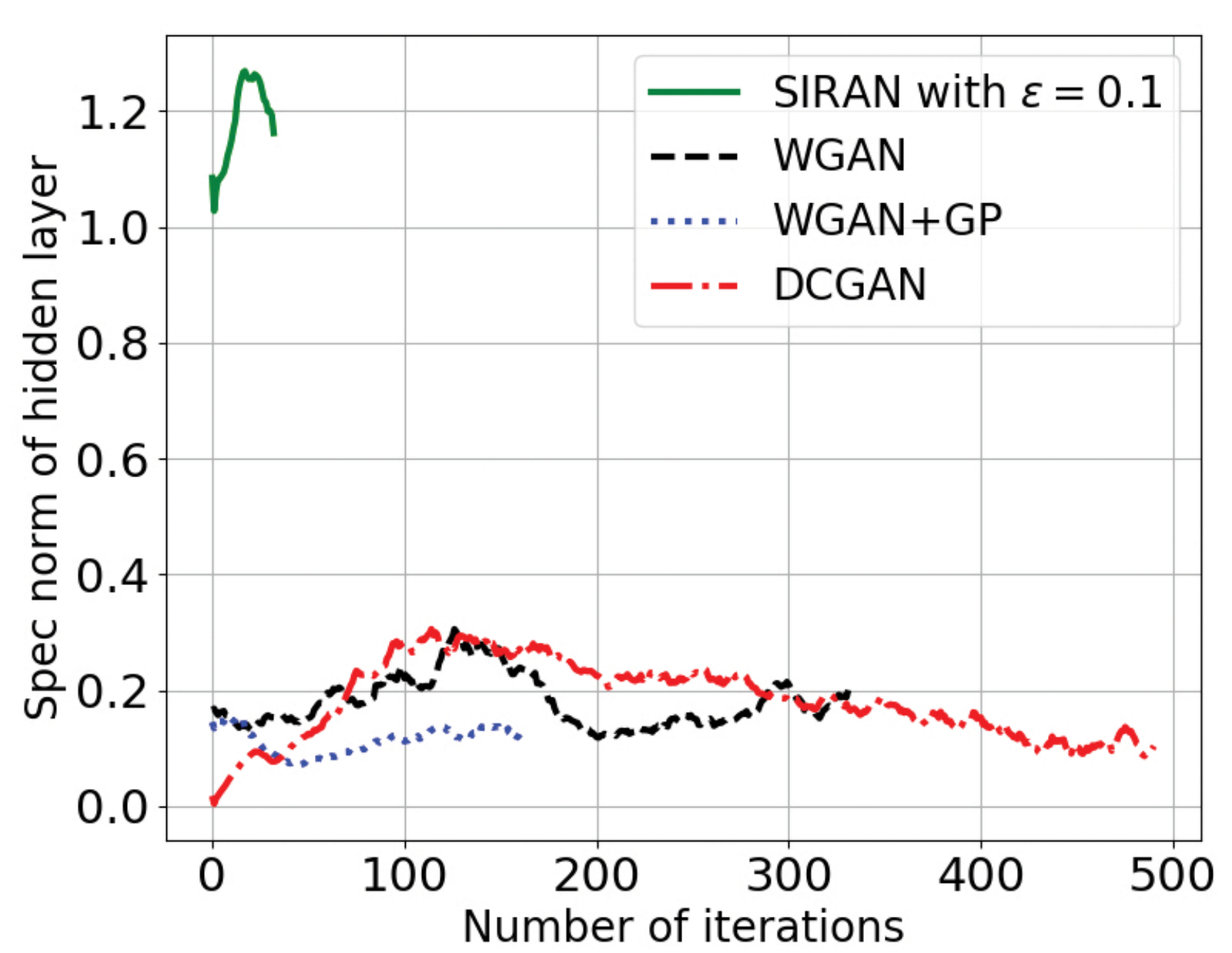}
  \caption{Approximated Spectral norm of gradients of hidden layer for different learning methods}\label{fig777}
  \endminipage
\end{figure}

\bibliographystyle{splncs04}
\bibliography{main}

\bibliographystyle{splncs04}
\bibliography{Main}

%




\end{document}